\documentclass[a4paper,11pt]{article}
\usepackage{amsmath}
\usepackage{amssymb}
\usepackage[dvips]{graphicx}
\usepackage{graphicx}
\usepackage{fullpage}
\usepackage{appendix}
\usepackage{hyperref}
\usepackage{color}  
\usepackage{times}
\usepackage{leftidx}
\usepackage{scrextend}
\usepackage{fixmath}
\usepackage{float}
\usepackage{authblk}

\newcommand{\remove}[1]{}

\newtheorem{myclaim}{Claim}%[section]

\newtheorem{proposition}{Proposition}%[section]

\newtheorem{corollary}{Corollary}%[section]

\newtheorem{definition}{Definition}%[section]

\newtheorem{lemma}{Lemma}%[section]

\newtheorem{problem}{Problem}%[section]

\newtheorem{theorem}{Theorem}%[section]

\newtheorem{remark}{Remark}%[section]

\newenvironment{proof}{\noindent{\bf Proof\@:}}{\hfill $\Diamond$\\}
\newenvironment{theoremproof}[1]{\noindent{\bf Proof of Theorem #1\@:}}{\hfill $\Diamond$\\}

\newenvironment{propositionproof}[1]{\noindent{\bf Proof of Proposition #1\@:}}{\hfill $\Diamond$\\}

\newcommand\G{\mathbold{G}}
\newcommand\SG{\mathbold{H}}

\date{\today}

\title{A simple algorithm for sampling colourings of\\ $G(n,d/n)$ up to Gibbs Uniqueness Threshold
\footnote{Parts of this  work appeared  in SODA 2012 \cite{mySampling} and 
ESA 2014 \cite{mySamplingUp2Uniqueness}.}}
\author{
Charilaos Efthymiou\footnote{This work is  supported by  Deutsche Forschungsgemeinschaft (DFG)
grant  EF 103/11}\\ 
Goethe University\\
Frankfurt am Main, 60325, Germany \\
{\tt E-mail}: efthymiou@gmail.com
}

\begin{document}
\maketitle
\pagestyle{plain}

\begin{abstract}

\noindent
Approximate random $k$-colouring of a graph $G$ is a  well studied problem in computer science and 
statistical physics.  It amounts to constructing a $k$-colouring of $G$ which is distributed  close to 
{\em Gibbs distribution} in polynomial time.  Here, we deal with the problem when the underlying graph 
is an instance of Erd\H{o}s-R\'enyi random graph $G(n,d/n)$, where $d$ is a sufficiently large constant.

We propose a novel efficient algorithm for approximate random $k$-colouring
$G(n,d/n)$ for any $k\geq (1+\epsilon)d$.  To be more specific, with probability at least $1-n^{-\Omega(1)}$ over the
input instances $G(n,d/n)$ and for $k\geq (1+\epsilon)d$,  the algorithm returns a 
$k$-colouring which is distributed within total variation distance $n^{-\Omega(1)}$ from
the Gibbs distribution of the input graph instance.

The algorithm we propose is neither a MCMC one nor  inspired by the message passing algorithms proposed
by statistical physicists.   Roughly the idea is as follows: Initially we remove sufficiently many
edges of the input graph. This results in a ``simple graph" which can be $k$-coloured  randomly efficiently. 
The algorithm colours randomly this simple graph.
Then it puts  back the removed edges one by one.   Every time  a new  edge is put back the
algorithm  updates the  colouring of the graph  so that the colouring remains  random.

The performance of the algorithm depends heavily on certain  spatial correlation decay 
properties of the Gibbs distribution.

\medskip
\medskip
\noindent
{\bf Key words:} Random colouring, sparse random graph, efficient algorithm.\\ \vspace{-.3cm}

\noindent
{\bf AMS subject classifications:} Primary 68R99, 68W25,68W20 Secondary: 82B44
\end{abstract}

\newpage

\section{Introduction}\label{sec:intro}	

\noindent{\em
Let $\G=G(n,d/n)$ denote the random graph on the vertex set $V(\G)=\{1,\ldots,n\}$ where each edge
appears independently with probability $d/n$, for a sufficiently large fixed number $d>0$.
}
\\

Approximate random $k$-colouring of a graph $G$ is a   well studied problem. 
It amounts to constructing a $k$-colouring of $G$ which is distributed close to 
{\em Gibbs distribution}, i.e. the uniform distribution over all the $k$-colourings of $G$, in polynomial
time. Here, we consider the problem when the underlying graph is an instance of Erd\H{o}s-R\'enyi
random graph $\G=G(n,d/n)$. This  problem is a rather  natural  one  and it has gathered focus in 
computer science but also in   statistical physics.

From a technical perspective,  the main challenge is to  deal with the so called  {\em effect of high degree}
vertices.
That is, there is a relative large fluctuation on the degrees in  $\G$.  E.g. it is elementary to verify that the 
typical  instances of $\G$ have maximum degree $\Theta\left(\frac{\log n}{\log \log n}\right)$, while in these
instances more  than  $1-e^{-O(d)}$  fraction of  the vertices  have  degree in the interval $(1\pm \epsilon)d$.
Usually the bounds for  sampling $k$-colourings w.r.t. $k$ are expressed it terms of the {\em maximum  degree} 
e.g. \cite{RCGeneralBound,COLPaperA,RCLocalySparse,RCLattice,MartinelliTrees}. However, for $\G$ it is
natural to have bounds for $k$ expressed in terms of the {\em expected degree} $d$, rather than the maximum
degree.

The related work on this problem can be divided into  two strands. The first one 
is based on {\em Markov Chain Monte Carlo} (MCMC) approach.
There,  the goal is to prove that some appropriately defined Markov Chain\footnote{e.g.  Glauber dynamics} over the 
$k$-colourings of the input graph  is rapidly mixing.   The MCMC approach
to the problem is well studied   \cite{myMCMC,old-GnpSampling,mossel-colouring-gnp}.
The  most recent of these works,  i.e. \cite{myMCMC}, shows that  the well known Markov chain 
{\em Glauber block dynamics}  has polynomial  mixing time for  typical instances of $\G$ as long as the number  of 
colours $k\geq \frac{11}{2}d$. This is the lowest bound for $k$ as far as MCMC sampling is concerned.

The second strand has been based on message passing algorithms such as \emph{Belief propagation} 
\cite{SP-heuristic},  which are closely related to the (non-rigorous) statistical mechanics  techniques for the analysis 
of the random graph colouring problem.   These message passing algorithms aim to approximate (conditional) \emph{marginals}
of the Gibbs distribution at each vertex .  Given the marginals, a colouring  can be sampled by choosing 
a vertex $v$, assigning it a random colour  $i$ according to the marginal distribution, and repeating the procedure
with the colour of $v$ fixed to $i$. Of course, the challenge is to prove that the algorithm does indeed yield sufficiently 
good estimates of the marginals. 
In a similar spirit, and subsequently to this work,
 the authors of  \cite{YinFPAUS} propose  an approximate random colouring algorithm for 
$\G$ which uses the so-called Weitz's computational tree approach, from \cite{Weitz}, to compute Gibbs marginals 
for colorings. 
This algorithms requires at least $3d$ many colours for the running time to be polynomial, i.e. $O(n^s)$ for some  $s=s(d)>0$.

In this work we obtain a considerable improvement over the best previous results by presenting a novel algorithm that 
only requires $k=(1+\epsilon)d$ colours. The new algorithm does not fall into any of the  categories discussed 
above.  Instead, it rests on the following approach:
Given the input graph, first remove  sufficiently many vertices such that the resulting graph has a ``very simple" 
structure and it can be randomly $k$-coloured {\em efficiently}.  Once we have a random colouring of this, simple, 
graph we  start adding one by one all the edges we have  removed in the first place. 
Each time we put back in the graph an edge we  {\em update}  the colouring so that 
the new graph remains (asymptotically)  randomly coloured. Once the algorithm has  rebuilt the initial graph 
it returns its  colouring.

Perhaps the most challenging part of the algorithm is   to  {\em update} the colouring once we 
have added an extra edge. The problem can be formulated as follows.
Consider two fixed graphs $G$ and $G'$  such that $V(G)=V(G')$ and $E(G')=E(G)\cup \{v,u\}$
for some $v,u\in V(G)$. Given $X$, a random $k$-colouring of $G$, we want to create {\em efficiently} 
a  random $k$-colouring of the slightly more complex graph $G'$.
It is easy to show that if the vertices  $v,u$  have different colour  assignments under 
$X$, then   $X$ is a random $k$-colouring of $G'$.  
The interesting case is when $X(v)=X(u)$.  Then the algorithm should alter the colour assignment of
at least one of the two vertices such that the resulting colouring  is random conditional that
the assignments of $v$ and $u$ are different. Here, we use  an operation which we call  ``switching" 
so as to alter the colouring of only one of the two vertices. 
Roughly speaking, the switching chooses an appropriately large part of $G$, which 
contains only $v$. Then, it repermutes appropriately the colour classes in  this part of 
$G$ so as to get  the updated colouring.

For presenting our results we    use the notion of {\em total variation distance}, which is  a measure
of distance between distributions.
\begin{definition}
For the distributions $\nu_{a}, \nu_{b}$ on $[k]^V$, let 
$|| \nu_{a}-\nu_{b}||$ denote their {\em total variation distance},
i.e.
\begin{displaymath}
|| \nu_{a}-\nu_{b}||=\max_{\Omega' \subseteq [k]^V} | \nu_{a}(\Omega')-\nu_{b}(\Omega')|.
\end{displaymath}
For $\Lambda \subseteq V$ let $||\nu_{a}-\nu_{b}||_{\Lambda}$
be  the total variation distance between the projections of 
$\nu_a$ and $\nu_{b}$ on $[k]^{\Lambda}$.
\end{definition}
\begin{theorem}\label{thrm:GnpAccuracy}
Let $\epsilon>0$ be a fixed number, let $d$ be sufficiently large number and 
fixed $k\geq  (1+\epsilon)d$.
Consider    $\G=G(n,d/n)$ and let $\mu$ the uniform distribution
over the $k$-colouring of $\G$. Let $\hat{\mu}$ be the distribution of the colouring that 
is returned  by our algorithm on input $\G$.

Let $c=\frac{\epsilon}{80(1+\epsilon/4)\log d}$,   with probability at least $1-n^{-c}$  over the input  instances $\G$ it holds that 
\begin{equation}\label{eq:thrm:GnpAccuracy}
||\mu-\hat{\mu}|| = O\left( n^{-c}\right).
\end{equation}
\end{theorem}
The proof of Theorem \ref{thrm:GnpAccuracy} appears in Section \ref{sec:AppGnpNew}.

The following theorem  is for the time complexity of the algorithm,   its proof appears 
in Section \ref{sec:AppGnpNew}.
\begin{theorem}\label{thrm:GnpTimeCmplxt}
With probability at least $1-2n^{-2/3}$ over the input  instances $\G$, 
the time complexity of the random colouring algorithm is $O(n^{2})$.
\end{theorem}
Whether the running time of the algorithm is polynomial or not, depends on certain 
structural properties of the input graph $\G$. 
Mainly, these properties require that the ``short cycles" of $\G$ are disjoint.
 It will be  trivial to  distinguish the instances that can be coloured randomly efficiently 
by our algorithm from those that cannot, see in Section \ref{sec:AppGnpNew} for further details.

\begin{remark}
The region of $k$ for which our algorithm  operates,  coincides with what is 
conjectured  to be the so-called  ``{\em Uniqueness phase}" of the $k$-colourings of $\G$,
e.g. see \cite{ZdKrzk}. 
\end{remark}

\paragraph{Remarks on the accuracy}
Typically, the approximation guarantees we get from algorithms  as those in \cite{myMCMC,YinFPAUS} express 
the running  time of the algorithm as a polynomial of the error in the output. 
The running time and the error of the algorithm here are independent,
in the sense  that  the approximation guarantees do not improve by allowing the algorithm 
run more steps.

\paragraph{Notation} 
Given some graph $G$, we let $V(G)$ and $E(G)$ denote the vertex sets and the edge set, respectively.
Also, we let $\Omega_{G,k}$  be the set of  proper $k$-colourings of $G$. We denote with small letters 
of the greek alphabet  the colourings in $\Omega_{G,k}$, e.g. $\sigma, \eta, \tau$. We use capital letters for the random
variables which take values over the colourings e.g. $X,Y, Z$.  We denote with $\sigma_v, X(v)$ the colour assignment
of the vertex $v$ under the colouring $\sigma$ and $X$, respectively.  Given some  $\sigma\in \Omega_{G,k}$, for every
$i\in [k]$ we let $\sigma^{-1}(i)\subseteq V(G)$ be the colour class of colour $i$ under the colouring $\sigma$.
Finally, for some  integer $h>0$, we let $[h]=\{1,\ldots,h\}$.

\section{Basic Description}\label{sec:basics}

So as to give a basic description of our algorithm, we need to introduce few notions.
Consider  a fixed graph $G$ and let $v$ be a vertex in $V(G)$. 
Let $c,q\in [k]$ be different with each other and   let $\sigma$ be a $k$-colouring  of $G$
such that $\sigma(v)=c$.
We call {\em disagreement graph} $\mathbold{Q}=\mathbold{Q}(G,v,\sigma,q)$,   the maximal,
connected, induced subgraph of $G$ such that $v\in V(\mathbold{Q})$,  while 
$V(\mathbold{Q}) \subseteq \sigma^{-1}(c)\cup \sigma^{-1}(q)$.
\begin{remark}
The concept of  disagreement graph,  in the graph theory literature is  also known as 
Kempe Chain. 
\end{remark}
In Figure \ref{fig:G0},  the disagreement graph  $\mathbold{Q}(G,v,\sigma, ``{\tt green}")$ is the one 
with the fat lines. Note that $\sigma$ specifies a two colouring for the vertices of 
$\mathbold{Q}(G,v,\sigma, ``{\tt green}")$.
\begin{definition}
Consider  $G$, $v$,  $\sigma$ and $q$ as specified above,
as well as the disagreement graph $\mathbold{Q}=\mathbold{Q}(G,v,\sigma,q)$.
The ``$q$-switching of $\sigma$" corresponds to the  colouring of $G$
which is derived by exchanging the assignments in the two colour classes 
in $\mathbold{Q}$.
\end{definition}

\noindent
Figure \ref{fig:G1} illustrates a switching of the colouring in Figure \ref{fig:G0}.
That is, the colouring in Figure \ref{fig:G1} differs from the one in Figure \ref{fig:G0}
in  that we have exchanged the two colour classes of the subgraph with the fat lines.
The $q$-switching of any proper colouring of $G$ is always 
a proper colouring, too.

\begin{figure}
\begin{minipage}{0.5\textwidth}
	\centering
		\includegraphics[height=2.5cm]{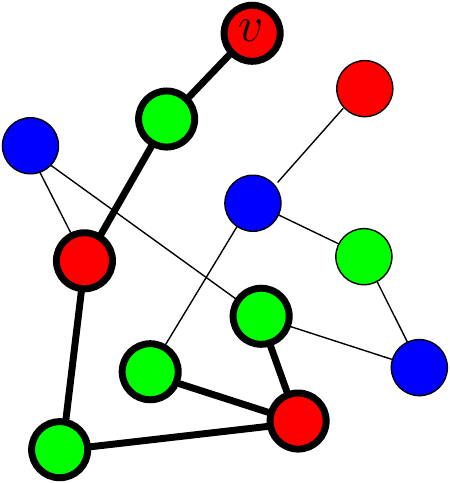}
		\caption{``Disagreement graph''.}
	\label{fig:G0}
\end{minipage}
\begin{minipage}{0.5\textwidth}
	\centering
		\includegraphics[height=2.5cm]{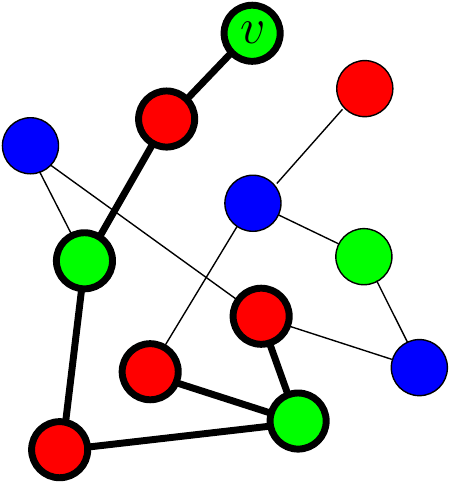}
		\caption{``switching''.}
	\label{fig:G1}
\end{minipage}	
\end{figure}

We proceed with a high level description of the algorithm. The input  is  $\G=G(n,d/n)$  and some integer 
$k\geq (1+\epsilon)d$.  The algorithm is as follows:\\ \vspace{-.2cm}

\noindent
{\bf Set up:}
We construct a sequence  of graphs $G_0,\ldots,G_r$ such that 
$G_{r}$ is  identical to $\G$ and  $G_i$ is a subgraph of  $G_{i+1}$.  
Each $G_i$ is derived by deleting from $G_{i+1}$ the edge $\{v_{i}, u_i\}$.
This edge is chosen at random among those which do not belong to a {\em short cycle} of $G_{i+1}$. 
We call short,  any cycle of length less  than $(\log_d n)/9$. $G_0$ is the graph  we get  when there
are  no other  edges to delete.
\\ \vspace{-.3cm}

\noindent
With probability $1-n^{-\Omega(1)}$, over the instances of $\G$,    the above process generates $G_0$ 
which is simple\footnote{In our case,   $G_0$ is considered simple if it is component structure is as  follows:
Each component is either an isolated vertex, or  a simple isolated cycle.  In Section \ref{sec:AppGnpNew}
we describe how someone can get efficiently a random colouring of such a graph.}  enough  that can be
$k$-coloured  randomly  in  polynomial time. If $G_0$ is not simple, the algorithm cannot proceed and abandons.
Assuming that   $G_0$ is simple, the algorithm  proceeds  as follows: 
\\ \vspace{-.2cm}

\noindent
{\bf Update:}
Take a random colouring of $G_0$.  Let $Y_0$ be that colouring. We get $Y_1,Y_2,\ldots, Y_r$, the colourings
of $G_1,G_2,\ldots,G_r$, respectively, according to the following inductive rule: Given that $G_i$ is coloured
$Y_{i}$,  so as to get $Y_{i+1}$ we  distinguish two cases
\begin{description} 
\item[Case (a):]  $Y_i$ (the colouring of $G_i$) assigns  $v_i$ and $u_i$  different colours, i.e. 
$Y_i(v_i)\neq Y_i(u_i)$
\item[Case (b):]  $Y_i$ assigns $v_i$ and $u_i$  the same
colour, i.e. $Y_i(v_i) = Y_i(u_i)$.
\end{description}
In the first case, we  set $Y_{i+1}=Y_i$, i.e. $G_{i+1}$ gets the same colouring as $G_i$. In the second case,  
we choose $q$ uniformly at random from $[k]\backslash\{Y_i(v_i)\}$, i.e.  among all the colours  but $Y_{i}(v_i)$. 
Then, we set $Y_{i+1}$ equal to the $q$-switching of $Y_{i}$.  The $q$-switching is w.r.t. the graph $G_i$, the 
vertex $v_i$ and the colouring $Y_i$. The algorithm repeats these steps for $i=0,\ldots, r-1$. Then it outputs $Y_r$.
\\ \vspace{-.2cm}

\begin{figure}
\begin{minipage}{0.5\textwidth}
	\centering
		\includegraphics[height=3cm]{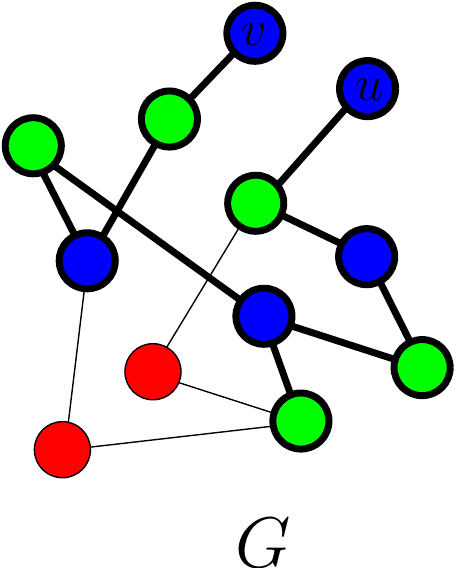}
		\caption{``Disagreement graph''.}
	\label{fig:G0Path}
\end{minipage}
\begin{minipage}{0.5\textwidth}
	\centering
		\includegraphics[height=3cm]{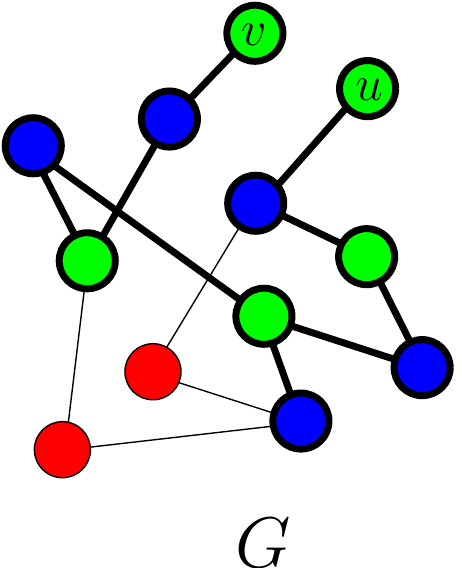}
		\caption{``switching''.}
	\label{fig:G1Path}
\end{minipage}	
\end{figure}

\noindent
One could remark that the switching does not necessarily provide a $k$-colouring where the  assignments 
of $v_i$ and $u_i$ are  different. That is,  it may be that both vertices $v_i, u_i$ belong  to the disagreement
graph in $Y_i$,  e.g. Figure \ref{fig:G0Path}.  Then,  after the $q$-switching  the colour assignments of $v_i$ 
and $u_i$  remain the same,  e.g. Figure \ref{fig:G1Path}. It turns out that  this situation is rare as long as
$k= (1+\epsilon)d$. More specifically,  with probability $1-o(n^{-1})$, the $q$-switching of $Y_i$ specifies 
different colour assignments for $v_i, u_i$.

The approximate nature of the algorithm amounts exactly to the fact that  on some, rare, occasions the 
switching somehow fails. The error at the output of the algorithm (see  Theorem \ref{thrm:GnpAccuracy})
is closely related  to the probability  of the event that our algorithm  encounters such  failure when the input is a typical 
instance of $\G$.
\begin{remark}
The lower bound we have for $k$ depends exactly how well we can control these failures of 
switching. That is,  for $k\leq d$  our analysis cannot guarantee that  the switching fails 
only on rare occasions.
\end{remark}

\section{The setting for the analysis of the algorithm.}\label{sec:InfSetting}

Consider a fixed graph $G$ and let $v,u$ be two distinguished, 
non-adjacent, vertices. 
\begin{definition}[Good \& Bad colourings] 
 Let $\sigma$ be a proper $k$-colouring of $G$, for some $k>0$.  We call $\sigma$  {\em bad  colouring}  
 w.r.t. the vertices $v,u$ of $G$,  if $\sigma_{v}=\sigma_{u}$. Otherwise, we call $\sigma$ {\em good}. 
\end{definition}
The idea that underlies the sampling algorithm, reduces the sampling problem  to dealing with the 
following one.
\begin{problem}\label{prblm:StepProblem}
Given a {\em bad}   random colouring of $G$, w.r.t. $\{v,u\}$, turn it to  a {\em good}  random colouring,
in  polynomial time. 
\end{problem}
Consider two different  $c,q\in [k]$ and let $\Omega_{c,c}$ and $\Omega_{q,c}$ be the set of 
colourings of $G$ which assign the pair of vertices $(v,u)$ colours $(c,c)$ and $(q,c)$, respectively. Our
approach to  Problem \ref{prblm:StepProblem} relies on getting  a mapping $H_{c,q}:\Omega_{c,c}\to \Omega_{q,c}$ 
such that the following holds:
\begin{description}
\item[A.] If $Z$ is uniformly  random  in $\Omega_{c,c}$,  then $H_{c,q}(Z)$ is  uniformly random in $\Omega_{q,c}$
\item [B.] The computation of $H_{c,q}(Z)$ can be accomplished in polynomial time.
\end{description}
It is straightforward  that having such a mapping for every two $c,q\in[k]$, it is sufficient  to solve Problem \ref{prblm:StepProblem}.
In the following discussion our focus is on (the more challenging) ${\bf A.}$ rather than ${\bf B.}$

An ideal (and to a great extent untrue)  situation  would have  been if $\Omega_{c,c}$ and  
$\Omega_{q,c}$ admitted a bijection. Then for ${\bf A.}$ it would suffice to use  for   $H_{c,q}$ 
 a bijection  between the two sets.  
Since this  is not expected to hold in general, our approach is  based on introducing an 
{\em approximate bijection} between  the sets  $\Omega_{c,c}$  and  $\Omega_{q,c}$.  That is,
we consider a mapping which is a bijection  between two sufficiently large subsets  of $\Omega_{c,c}$ 
and  $\Omega_{q,c}$, respectively. This would mean that if $Z$ is uniformly
random in $\Omega_{c,c}$ and $H_{c,q}(\cdot)$ an approximate bijection
between   $\Omega_{c,c}$  and  $\Omega_{q,c}$,  then  $H_{c,q}(Z)$ is {\em approximately} uniformly random in $\Omega_{q,c}$.

To be more specific, we let  $H_{c,q}$ represent the operation of $q$-switching  over the colourings
in $\Omega_{c,c}$, as we describe in Section \ref{sec:basics}.  For such mapping, we can find appropriate
 $\Omega'_{c,c}\subseteq \Omega_{c,c}$ and  $\Omega'_{q,c}\subseteq \Omega'_{q,c}$ such that
$H_{c,q}$ is a bijection between the sets $\Omega_{c,c}\backslash \Omega'_{c,c}$ and  
$\Omega_{q,c}\backslash \Omega'_{q,c}$.  
We call {\em pathological} each colou\-ring $\sigma\in \Omega'_{c,c}\cup   \Omega'_{q,c}$.  
For the pathological colouring $\sigma\in \Omega'_{c,c}$ it holds that $H_{c,q}(\sigma)\notin \Omega_{q,c}$,
while  for $\sigma\in \Omega'_{q,c}$ it holds that $H^{-1}_{c,q}(\sigma)\notin \Omega_{c,c}$. 
\begin{remark}
There is
a natural characterization for  the pathological colourings $\sigma\in \Omega_{c,c}$. That is, $\sigma$ is pathological if the 
disagreement graph $\mathbold{Q}=\mathbold{Q}(G,v,\sigma,q)$ contains both $v,u$.
\end{remark}
It turns out that,  for $Z$ being  uniformly random in $\Omega_{c,c}$,  $H_{c,q}(Z)$ is distributed within 
total variation distance  $\max\left\{\frac{\Omega'_{c,c}}{\Omega_{c,c}}, \frac{\Omega'_{q,c}}{\Omega_{q,c}}\right\}$
from  the uniform distribution over $\Omega_{q,c}$. That is, the error we introduce with the approximate
bijection $H_{c,q}$ depends on the {\em relative  number} of the pathological colorings in $\Omega_{c,c}$ 
and $\Omega_{q,c}$, respectively. 
A key ingredient of our analysis is  to provide  appropriate upper bounds for the two ratios  
$\Omega'_{c,c}/\Omega_{c,c},  \Omega'_{q,c}/\Omega_{q,c}$.

\subsection{Bounding the Error - Spatial Mixing}\label{sec:PrelBoundErr}

As in the previous section, let  $G$ be fixed. 
For bounding the ratios $\Omega'_{c,c}/\Omega_{c,c}$ and $\Omega'_{q,c}/\Omega_{q,c}$,
we treat both cases in the same way, so let us focus on bounding $\Omega'_{c,c}/\Omega_{c,c}$.

It is direct that   $\Omega'_{c,c}/\Omega_{c,c}$ expresses the probability of getting a pathological 
colouring if we choose uniformly at random from  $\Omega_{c,c}$. 
For this, consider the  situation  where we choose 
u.a.r. from $\Omega_{c,c}$.  For every path  $P$ that connects $v,u$ in the graph $G$, we let 
$\mathbf{I}_{\{P\}}$  be an indicator variable which is one if the vertices in the path $P$ are coloured
only with colours $c,q$ in the random colouring   and zero otherwise. 
Equivalently, $\mathbf{I}_{\{P\}}=1$ if and only if $P$ belongs to the  graph of disagreement 
that is induced by the random colouring and the colour $q$.
It holds that 
\begin{equation}\label{eq:OmegaRatioSumOfIndicators}
\frac{\Omega'_{c,c}}{\Omega_{c,c}}=\Pr\left[ \sum_{P}\mathbf{I}_{\{P\}}\geq 1\right] 
\leq \sum_{P}\Pr\left[\mathbf{I}_{\{P\}}=1\right].
\end{equation}
The first equality follows from the fact that if both $v,u$ belong to the disagreement graph, then
there should be  at least one path $P$  such that $\mathbf{I}_{\{P\}}=1$. The last inequality follows from 
the union bound.

\begin{figure}
	\centering
		\includegraphics[height=2.3cm]{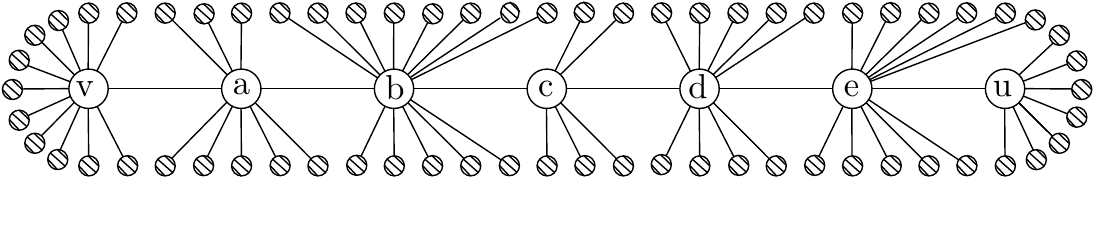}
		\caption{Boundary at distance 1 from the path.}
	\label{fig:OldAppr}
\end{figure}

\begin{remark}
The above inequality bounds the relative number of pathological colourings in  ${\Omega_{c,c}}$ 
(resp. in $\Omega_{q,c}$ ) with the  expected number of paths from $v$ to $u$  which are coloured 
with $c,q$ under a colouring  which is chosen at random from $\Omega_{c,c}$ (resp. $\Omega_{q,c}$).
\end{remark}

\noindent
In general,  computing $\Pr[\mathbf{I}_{\{P\}}=1]$ exactly is a formidable task to accomplish 
due to the complex structure we typically  have in the underlying graph.   For this reason we reside on computing  
upper bounds  of this probability term.  

In  \cite{mySampling} we used the idea of the so-called  
``Disagreement percolation" from \cite{disperc}. The  setting of this approach is illustrated in Figure 
\ref{fig:OldAppr},
for the path $P=(v, a,b,c,d,e, u)$. The lined vertices  are exactly 
these  which are adjacent to the path.  So as to bound the probability that the path $P$ is coloured
with $c,q$, we assume a worst case  boundary colouring for the lined vertices.
Given the fixed colourings  at the boundary, we take a random
colouring of the uncoloured vertices in $P$, conditional $v,u$ are assigned $c$,
and estimate the probability that $P$ is coloured exclusively with $c,q$.   
\begin{remark}
The choice of the boundary  above, is worst case in the sense that it maximizes the probability that 
 $\mathbf{I}_{\{P\}}=1$.
\end{remark}

\noindent
It turns out that considering the worst case boundary condition next to the path $P$ is a too pessimistic 
assumption.  There is an improvement  once we  adopt  a less restrictive
approach.  The  new approach is illustrated in Figure \ref{fig:NewAppr}. Roughly speaking, 
we  consider a worst case boundary condition at  the vertices around $P$ which are at graph  distance
$r$, for $r \gg 1$.  The boundary condition gives rise to Gibbs distribution 
over the $k$-colourings of the subgraph confined by the boundary vertices. 
In particular, we argue about
the  {\em spatial mixing} properties of the Gibbs distributions in the confined graph. 
We show  that the colouring\footnote{any colouring} of the distant  vertices does not bias the distribution of the colour assignment of the vertices in $P$ by too much.

The above approach is well motivated when we consider $G(n,d/n)$.  For such graph, typically,   around 
 most of the vertices in $P$ we  have a tree-like neighbourhood
of maximum degree  very close to the expected degree $d$.  This  gives 
rise to study correlation decay for random colourings of a tree with maximum degree $\Delta$, 
for   $\Delta\approx d$. Our spatial mixing results build on the  work of Jonasson 
\cite{TreeUniq}.
\begin{figure}
	\centering
		\includegraphics[height=2.5cm]{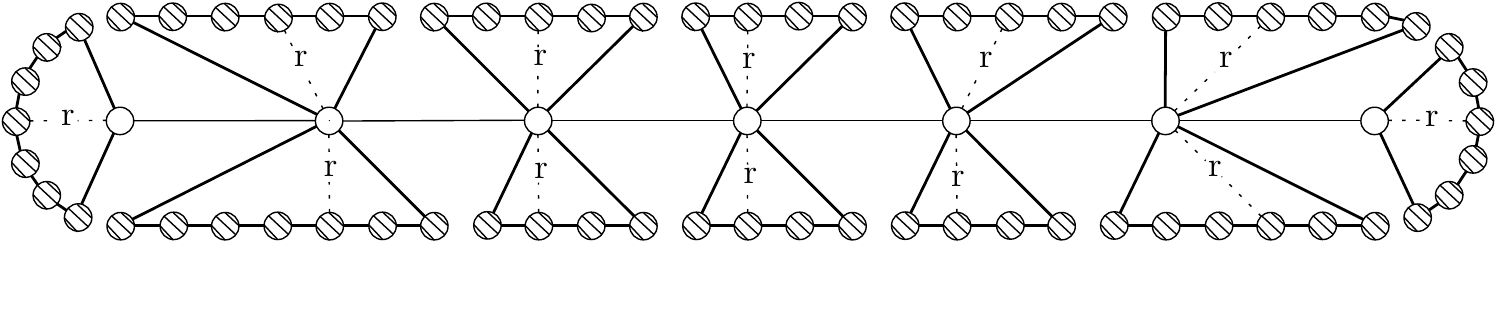}
		\caption{Boundary at distance $r$ from the path}
	\label{fig:NewAppr}
\end{figure}
\paragraph{From fixed graph to random graph.}
When the underlying graph $G$ is fixed, 
we bound  $\Omega'_{c,c}/\Omega_{c,c}$  (resp. $\Omega'_{q,c}/\Omega_{q,c}$)  by using  the  
expected number of paths between $v$ and $u$ that are coloured $c,q$ in a colouring chosen 
uniformly at random from $\Omega_{c,c}$ (resp. $\Omega_{q,c}$). That is, we need to argue 
on the randomness of the $k$-colourings of $G$.

In our analysis, we  deal with cases where the underlying graph is random. Then,  we have
an extra  level of randomness to deal with, that of the graph instance. That is, we  take an instance of the graph
and then, given the graph, we consider a random colouring of this graph instance. Even in this setting, 
we  compute the expected number number of paths between $v$ and $u$ that are coloured $c,q$,  
however,  the expectation is w.r.t. to the randomness of both the  graph and its colouring.  A result which 
is central in our analysis is the following one.
\begin{theorem}\label{thrm:PrIPGnp}
Let $\epsilon>0$,  let $d>0$ be sufficiently large and let fixed  $k\geq (1+\epsilon)d$. 
Consider  $\mathbold{G}=G(n,d/n)$.  Let the graph $\mathbold{H}$ be such that 
$V(\mathbold{H})=V(\mathbold{G})$ and $E(\mathbold{H})\subseteq E(\mathbold{G})$.
For any two $c,q\in [k]$, different with each other,  any non-negative integer $\ell\leq \log^2 n$ and a 
permutation $P=(w_0,\ldots, w_{\ell})$ of vertices in $V(\mathbold{H})$ the following is true:

Let $X$ be a random $k$-colouring of $\mathbold{H}$ conditional than $X(w_0)=c$. 
Let  $\mathbf{I}_{\{P\}}=1$, if  $P$ is a path in $\mathbold{H}$ and $X(w_i)\in \{c,q\}$, for every 
$j=1,\ldots, \ell$.  Otherwise $\mathbf{I}_{\{P\}}=0$. It holds that
\begin{equation}\label{eq:thrm:PrIPGnp}
\Pr[\mathbf{I}_{\{P\}}=1]\leq 2[(1+\epsilon/4)n]^{-\ell}.
\end{equation}
\end{theorem}
The proof of the theorem appears in Section \ref{sec:thrm:PrIPGnp}.
\begin{remark}
 In \eqref{eq:thrm:PrIPGnp} the probability term is w.r.t. both the randomness of $\mathbold{H}$
and the colouring $X$.
\end{remark}
The above theorem implies that for $k\geq (1+\epsilon)d$,  in a random $k$-colouring of $\mathbold{G}$,
typically,  there are not long paths coloured with only two colours. Furthermore, this property
is  monotone in the graph structure. That is, it holds even though if we remove an arbitrary 
number of edges from $\mathbold{G}$ (and get $\mathbold{H}$). The monotonicity property follows 
from the fact that we can extend  in a natural way the Gibbs uniqueness  condition in \cite{TreeUniq}
from $\Delta$ regular trees  to trees of maximum degree $\Delta$.

\section{Updating  Colourings}\label{sec:Problem1}
In this section, we describe  the process that  the random colouring
algorithm uses to update the colourings, we call it ${\tt Update}$. For the sake of clarity in this section we assume a 
fixed graph  $G$ and we distinguish two vertices $v,u\in V(G)$. We take $k$ sufficiently large so that $G$ is $k$-colourable. 
\begin{definition}[Disagreement graph]\label{def:DiagreementGraph}
For any  $\sigma\in \Omega_{G,k}$ and  $q\in [k]\backslash\{\sigma_v\}$ we let the disagreement
graph $\mathbold{Q}=\mathbold{Q}(G,v,\sigma,q)$ be the maximal induced subgraph of $G$ such that
\begin{displaymath}
V(\mathbold{Q})=\left \{x \in V(G) \left | 
\begin{array}{l}
\exists \textrm{ path  } w_1,\ldots,w_{\ell}, \textrm{ in $G$ such that:  } \\
w_1=v, w_{\ell}=x, \sigma(w_j)\in \{\sigma_v, q\}, \forall  j\in [\ell]
\end{array}
\right .
\right \}.
\end{displaymath}
\end{definition}

\noindent
Next, we provide the pseudo-code of the operation ${\tt Switching}$, presented in Section \ref{sec:basics}.
\\ \vspace{-.2cm}

\noindent
{\bf Switching} \\ \vspace{-.75cm} \\
\rule{\textwidth}{1.5pt}\\
{\bf Input:} $G, v$,  $\sigma$ and $q\in [k]\backslash\{\sigma_v\}$ \\
\hspace*{.6cm}{\tt set} $c=\sigma_v$\\ 
\hspace*{.6cm}{\tt set} $\mathbold{Q}=\mathbold{Q}(G,v, \sigma,q)$\\ 
\hspace*{.6cm}{\tt set} $\tau(V(G)\backslash V(\mathbold{Q}))=\sigma(V(G)\backslash V(\mathbold{Q}))$
\quad /* Everything outside $\mathbold{Q}$ keeps its initial colouring*/ \\
\hspace*{.6cm}{\tt  for} $w\in V(\mathbold{Q})\cap \sigma^{-1}(c)$ {\tt do}\\
\hspace*{1.2cm}{\tt set} $\tau(w)=q$\\
\hspace*{.6cm}{\tt  for} $w\in V(\mathbold{Q})\cap \sigma^{-1}(q)$ {\tt do}\\
\hspace*{1.2cm}{\tt set} $\tau(w)=c$\\
{\bf Output:} $\tau$\\ \vspace{-.75cm} \\
\rule{\textwidth}{1.5pt} \\  \vspace{-.2cm} 

\noindent
 ${\tt Switching}$ has the following property, whose proof is easy to derive.
\begin{lemma}\label{lemma:SwitchProperColour}
If  $\tau={\tt Switching}(G,v,\sigma,q)$, where  $\sigma\in \Omega_{G,k}$ and  $q\neq \sigma(v)$,
then  $\tau\in \Omega_{G,k}$.
\end{lemma}
The proof of Lemma  \ref{lemma:SwitchProperColour}, is quite straightforward and appears
in Section \ref{sec:lemma:SwitchProperColour}.

As far the time complexity of ${\tt Switching}$ is regarded we have the following lemma, whose
proof appears in Section \ref{sec:lemma:StepAc}.

\begin{lemma}\label{lemma:StepAc}
For every $v\in V(G)$, any   $\sigma \in \Omega_{G,k}$, $q\in [k]\backslash\{\sigma_v\}$
the time  complexity  of computing ${\tt Switching}(G,v, \sigma, q)$ is $O(|E(G)|)$.
\end{lemma}

\noindent
In what follows, we have the pseudo-code for  ${\tt Update}$. 
\\ \vspace{-.1cm}

\noindent
{\bf Update} \\ \vspace{-.75cm}\\
\rule{\textwidth}{1pt}
{\bf Input:} $G, v, u,$ $\sigma \in \Omega_{G,k}$ \\
\hspace*{.6cm} {\tt if }$\sigma$ is a  {\em good} colouring w.r.t. $v,u$, {\tt then}\\
\hspace*{1.2cm}{\tt set}  $\tau=\sigma$\\
\hspace*{.6cm} {\tt else do} \\
\hspace*{1.2cm} {\tt choose} $q$ u.a.r. from $[k]\backslash\{\sigma_v\}$\\
\hspace*{1.2cm} {\tt set} $\tau={\tt Switching}(G,v, \sigma,q)$ \\ 
{\bf Output:} $\tau$ \\ \vspace{-.75cm}\\
\rule{\textwidth}{1pt} \\ \vspace{-.7cm}\\

\noindent
To this end, we need argue about the time complexity and the accuracy of ${\tt Update}$. As far as the 
time complexity is regarded we have the following theorem.
\begin{theorem}\label{thrm:UpdateCmplxt}
For any $v,u\in V$,  $\sigma \in \Omega_{G,k}$ and $q\in [k]\setminus\{\sigma_v\}$, 
the time  complexity  of  ${\tt Update}(G,v, u, \sigma, k)$ is $O(|E(G)|)$.
\end{theorem}

\noindent
Theorem \ref{thrm:UpdateCmplxt} follows as a corollary of Lemma \ref{lemma:StepAc}, 
once we note that  the execution time of ${\tt Update}$ is dominated by  the calls of ${\tt Switching}$.

So as to study the accuracy of ${\tt Update}$  we  introduce the following
concepts. For any two different colours $c,q$ we let $S_{q}(c,c)\subseteq \Omega(c,c)$
and $S_{c}(q,c) \subseteq \Omega(q,c)$ be defined as follows: The set $S_{q}(c,c)$ (resp.  
$S_{c}(q,c)$) contains every $\sigma\in \Omega(c,c)$  (resp. $\sigma\in \Omega(q,c)$) 
such that  there is no path between $v$ and $u$ which is coloured  only with 
the colours $c,q$, by $\sigma$.
\begin{definition}\label{assm:isomorphism}
Let $\alpha=\alpha_{G,k}\in[0,1]$ be the minimum number such that the following holds:
For every pair of different colours $c,q\in [k]$  the sets  $S_{q}(c,c)$ and $S_{c}(q,c)$  contain 
all but an   $\alpha$-fraction  of colourings of  $\Omega(c,c)$ and $\Omega(q,c)$, respectively.
\end{definition}

\noindent
In general the value of $\alpha$ depends on the underlying graph $G$ and $k$. The quantity $\alpha$
is an upper bound on the relative size of pathological  colourings in each set $\Omega(c,c')$.  
\begin{theorem}\label{theorem:STEPAccuracy}
Let $\nu$ be the uniform distribution over the $k$-colourings of $G$ which are  {\em good}, 
w.r.t. $v,u$.  Let, also, $\nu'$ be the di\-stri\-bu\-tion of the output of ${\tt Update}$ when the
input colouring is distributed uniformly at random over the  $k$-colourings of $G$.  Letting
$\alpha$ be as in  Definition \ref{assm:isomorphism},  it holds that 
\begin{displaymath}
||\nu-\nu'||\leq \alpha.
\end{displaymath}
\end{theorem}
The proof of Theorem \ref{theorem:STEPAccuracy} appears in  Section \ref{sec:theorem:STEPAccuracy}.

\section{Random Colouring Algorithm}\label{sec:RCAlgorithm}
In this section, we  study the  time complexity and the accuracy of the random colouring algorithm.
 For the sake of definitiveness we assume the input graph $G$ to be fixed and 
is such that $G$ is $k$-colourable.
Given the input graph $G$, the algorithm creates  the sequence of
subgraphs $G_0, \ldots, G_r$.  The  variable $Y_i$ denotes the $k$-colouring 
that the algorithm assigns to the graph $G_i$.
$G_{i}$ is derived by deleting from $G_{i+1}$ an edge which we call   $\{v_i, u_i\}$.     

As we consider a general graph $G$,  in the pseudo-code that follows,
we do not specify  exactly how do we get $G_i$  from $G_{i+1}$, 
i.e. what is $\{v_i,u_i\}$.  Also, we do not specify how do we get   $Y_0$, the random colouring of $G_0$. 
We  get specific on these two matters only when we consider $G(n,d/n)$ at the input,
see Section \ref{sec:AppGnpNew}.

The pseudo-code for the algorithm is as follows:
\\ \vspace{-.7cm}\\

\noindent
{\bf Random Colouring Algorithm}\\ \vspace{-.7cm} \\
\rule{\textwidth}{1pt} \\ 
{\bf Input:} $G$, $k$\\ 
\hspace*{.6cm}{\tt compute} $G_0,G_1\ldots, G_r$\\
\hspace*{.6cm}{\tt compute} $Y_0$  $\qquad \qquad\qquad /*$ {\em Get a random 
$k$-colouring of $G_0$}$*/$\\
\hspace*{.6cm}{\tt for} $0\leq i \leq r-1$ {\tt do} \\
\hspace*{1.2cm} {\tt set}  $Y_{i+1}$ the output of ${\tt Update}(G_i, v_i,u_i,Y_i, k)$\\
{\bf Output:} $Y_{r}$\\ \vspace{-.7cm} \\
\rule{\textwidth}{1pt} \\ \vspace{-.7cm} \\

\noindent
Using Theorem \ref{thrm:UpdateCmplxt} and noting that  $r\leq |E(G)|$, we get the following
result.
\begin{theorem}\label{thrm:TimeCmplxt}
Let $T(n)$ be the time complexity for $k$-colouring randomly  $G_0$.
Then, the random colouring algorithm has  time complexity $O(|E(G)|^2+T(n))$.
\end{theorem}

\noindent
Next,  we investigate  the accuracy of the algorithm.   For any  $c,q \in [k]$ we let $\Omega_i(c,q)$ be 
the set of colourings of $G_i$ which assign the colours $c$ and $q$ to the vertices $v_i$ and $u_i$, 
respectively.  Furthermore,  for two different colours $c,q\in [k]$,  let $S^i_q(c,c)\subseteq \Omega_i(c,c)$
and $S^i_c(q,c)\subseteq \Omega_i(q,c)$ be defined as follows: The set $S^i_q(c,c)$ (resp. $S^i_c(q,c)$) 
contains every $\sigma\in \Omega_i(c,c)$  (resp. $\sigma\in \Omega_i(q,c)$) such that there is no path 
between $v_i$ and $u_i$ (in $G_i$) which is coloured by $\sigma$ using  the colours $c,q$, only. 
\begin{definition}\label{def:SeqIsomorphism}
For every $i=0,\ldots,r-1$, let $\alpha_i\in [0,1]$ be the minimum number such that the following 
holds: For any pair of different colours $c,q$ the sets $S^i_q(c,c)$ and $S^i_c(q,c)$ contain all but
an  $\alpha_i$-fraction  of the colourings in  $\Omega_i(c,c)$ and $\Omega_i(q,c)$,  
respectively.
\end{definition}

\noindent
Clearly the quantities $\alpha_i$ depend on $G_i$ and $k$.
\begin{theorem}\label{thrm:Accuracy}
Let $\mu$ be the uniform distribution over the $k$-colourings of the input graph $G$. Let 
$\hat{\mu}$ be the distribution of the colourings at the output of the algorithm.  It holds that
\begin{displaymath}
||\mu- \hat{\mu}||\leq \sum^{r-1}_{i=0}\alpha_i,
\end{displaymath}
where $\alpha_i$ is from Definition \ref{def:SeqIsomorphism} and $r$ is the number of
terms of the sequence  $G_0,G_1,\ldots, G_r$.
\end{theorem}
The proof of Theorem \ref{thrm:Accuracy} appears in  Section \ref{sec:thrm:Accuracy}.

\section{Random Colouring $G(n,d/n)$}\label{sec:AppGnpNew}

In this section, we focus on the case where the input of  ${\tt Random\;   Colouring\;  Algorithm}$
is  $\mathbold{G}=G(n,d/n)$.  This study leads to the proof of Theorems \ref{thrm:GnpAccuracy}
and \ref{thrm:GnpTimeCmplxt}.

We start by describing how do we get  ${G}_0, \ldots, {G}_r$ from $\G$. 
Let ${\cal E}(\mathbold{G})\subseteq 
E(\mathbold{G})$  contain exactly every edge $e \in E(\mathbold{G})$ such that  the 
shortest simple cycle that contains $e$ is of length greater than  $(\log_d n)/9$.
\\ \vspace{-.1cm}

\noindent
{\bf Computing $G_0,\ldots,G_r$:}
The sequence  ${G}_0, \ldots, {G}_r$ is constructed as follows: 
Set $r=|{\cal E}|+1$.  We set ${G}_r=\mathbold{G}$.    Given ${G}_i$ we get ${G}_{i-1}$ by 
removing  a randomly chosen  edge of ${G}_i$ 
which also belongs to ${\cal E}(\G)$, for $i=1,\ldots, r$. 
${G}_0$   contains only the edges of the initial graph which do not belong to ${\cal E}(\G)$.
\\ \vspace{-.1cm}

\noindent
Perhaps it is interesting to describe what motivates the above construction of 
the sequence $G_0,\ldots,G_r$.
Since each $\alpha_i$ depends on $G_i$, we construct the sequence so
as to have $\sum_i\alpha_i$, as small as possible. 
The smaller  the probability the algorithm  encounters a disagreement graph which
includes both $v_i, u_i$  the  smaller  $\alpha_i$s  get. 
Choosing $v_i$ and $u_i$ to be at large distance reduces the probability that the disagreement
graph includes both of them, consequently, $\alpha_i$ gets smaller.
 Our choice of sequence forces $v_i$ and $u_i$ to be at distance greater than
$(\log_dn)/9$ with each other.  To a certain extent, this allows to control the error of the
algorithm, i.e.  $\sum_{i}\alpha_i$.

Given the sequence $G_0,\ldots,G_r$, the next step is  to argue on how can we get a random
$k$-colouring of  $G_0$, efficiently. Our arguments rely on the fact that typically
$G_0$ has a very simple structure, i.e. we use the following  result.
\begin{lemma}\label{lemma:seqsubprop}
For $d>0$, let $\mathfrak{S}_{n,d}$ be the set of all graph on $n$ vertices such that their component structure
is as follows: Each component is either the trivial\footnote{single isolated vertex}, or it is a simple isolated
 cycle
\footnote{the cycles do not share edges nor vertices}
of maximum length $(\log_d n)/9$.
Consider  $\G$ and the sequence ${G}_0,\ldots {G}_r$ created as we described above.
It holds that
$$
\Pr[G_0\in \mathfrak{S}_{n,d}]\geq 1-n^{-2/3}.
$$
\end{lemma}
The proof of Lemma \ref{lemma:seqsubprop} appears in Section \ref{sec:lemma:seqsubprop}. 

For  ${G}_0\in \mathfrak{S}_{n,d}$, {\em exact} random $k$-colouring can be implemented efficiently. 
In what follows we describe an efficient process that can colour randomly any graph in $\mathfrak{S}_{n,d}$.
\\ \vspace{-.1cm}

\noindent
{\bf Random Colouring in  $\mathfrak{S}_{n,d}$} \\ \vspace{-.7cm} \\
\rule{\textwidth}{1pt} \\ 
{\bf Input:} ${G}\in \mathfrak{S}_{n,d}$, $k$.\\ 
\hspace*{.6cm}{\tt set} ${\cal C}$ to be the set of all cycles in ${G}$\\
\hspace*{.6cm}{\tt for} each isolated vertex $v\in V({G})$ {\tt do} \hspace{4cm} /*Colouring isolated vertices*/\\
\hspace*{1.2cm}{\tt set} $\tau(v)$ a colour chosen uniformly random  from $[k]$ \\
\hspace*{.6cm}{\tt for} each $C=(w_0,\ldots w_l)\in {\cal C}$ {\tt do}  \hspace{4.5cm} /*Colouring isolated cycles*/\\
\hspace*{1.2cm}{\tt set} $\tau(w_0)$ a color chosen uniformly random  from $[k]$\\
\hspace*{1.2cm}{\tt for} $i=1,\ldots, l$ {\tt do}\\
\hspace*{1.8cm}{\tt set} $\mu_{w_i}$ the Gibbs marginal of $w_i$, conditional $\tau(w_0),\ldots,\tau(w_{i-1})$\\
\hspace*{1.8cm}{\tt compute} $\mu_{w_i}$  using Dynamic Programming\\
\hspace*{1.8cm}{\tt set} $\tau(w_i)$ according to $\mu_{w_i}$\\
{\bf Output:} $\tau$\\ \vspace{-.7cm} \\
\rule{\textwidth}{1pt} \\ \vspace{-.7cm} \\

\noindent
The most interesting part of the above algorithm is the one for random colouring of the cycles.
For each cycle $C\in {\cal C}$, the algorithm first  assigns a random colour on the vertex $w_0$. Once 
 $w_0$ is assigned a colour, then we eliminate the cycle structure of $C$ and  now we deal with a  tree of maximum 
 degree 2. This allows to compute the marginal $\mu_{w_i}$, for each vertex $w_i\in C$, by using 
{\em Dynamic Programming} (DP). %

\begin{remark}
The use of DP for computing Gibbs marginals on the trees is well known
to be exact, e.g. see \cite{GraphModels} for an excellent survey on the subject. 
\end{remark}

\begin{remark}
The recursive distributional equations that DP uses in this setting are more or less standard.  E\-xam\-ple of such equations
appear in the proof of Lemma \ref{lemma:spatial-monotone}, in Section \ref{sec:lemma:spatial-monotone}.
\end{remark}
Once  we get an exact random colouring of $G_0$ by using the above algorithm,  ${\tt Random\; Colouring\; Algorithm}$ 
colours the
remaining graphs ${G}_1,\ldots, {G}_r$  by using ${\tt Update}$, as we described in Section \ref{sec:RCAlgorithm}.

Let $\mathfrak{X}_{n,d}$ contain every  graph $G$ on  $n$ vertices such that the following holds:
\begin{enumerate}
\item getting a sequence of subgraphs  $G_0,\ldots, G_r$, as described in Section \ref{sec:AppGnpNew}, 
it holds that  $G_0\in \mathfrak{S}_{n,d}$
\item  $|E(G)|\leq (1+n^{-1/3})dn/2$.
\end{enumerate}
Note that for  some $G$ we have that  $G_0\in \mathfrak{S}_{n,d}$ regardless
of the order we remove the edges for creating the sequence $G_0,\ldots, G_r$.
That is,  whether $\G\in \mathfrak{X}_{n,d}$, or not, depends only on the graph
$\G$.

If the input graph $\G$ does not belong into $\mathfrak{X}_{n,d}$, then the ${\tt Random\; Colouring\; Algorithm}$ 
abandons. 
It turns out that this typically does not happen.  In particular, we have following corollary.

\begin{corollary}\label{cor:G(np)inXnd}
For sufficiently large $d>0$, it holds that $\Pr[\G\in \mathfrak{X}_{n,d}]\geq 1-2n^{-2/3}$.
\end{corollary}
\begin{proof} 
Lemma
\ref{lemma:seqsubprop}, states that  for the sequence $G_0,\ldots, G_r$ generated from $\G$ as described
in Section \ref{sec:AppGnpNew}  it holds that   $\Pr[G_0\in \mathfrak{S}_{n,d}]\geq 1-n^{-2/3}$. 
Using  Chernoff's 
bounds, e.g. \cite{janson}, we also get 
\[
\Pr\left [|E(\G)|\geq (1+n^{-1/3}){dn}/{2} \right]\leq \exp\left(-n^{1/4} \right).
\]
A simple union bound, yields that indeed $\Pr[\G\in \mathfrak{X}_{n,d}]\geq 1-2n^{-2/3}$.
\end{proof}

In the following two sections we prove Theorems \ref{thrm:GnpAccuracy} and \ref{thrm:GnpTimeCmplxt}.

\subsection{Proof of Theorem \ref{thrm:GnpAccuracy}}\label{sec:thrm:GnpAccuracy}

For proving Theorem \ref{thrm:GnpAccuracy} we need to use the following result,
whose proof appears in Section \ref{sec:thrm:eIsomGnp}.

\begin{theorem}\label{thrm:eIsomGnp}
Let $\epsilon,d,k$ be as in the statement of Theorem \ref{thrm:GnpAccuracy}.
Consider the sequence  $G_0,\ldots, G_r$  ge\-ne\-ra\-ted from $\G$ as described 
in Section \ref{sec:AppGnpNew}.
For any $i\in \{0,\ldots, r-1\}$ it holds that
\begin{displaymath}
\mathbb{E}[\alpha_i]\leq 50{\epsilon^{-1}k(4+\epsilon)} n^{-\left(1+\frac{\epsilon}{36(1+\epsilon/4)\log d} \right)}.
\end{displaymath}
\end{theorem}

\begin{theoremproof}{\ref{thrm:GnpAccuracy}}
In light of Corollary \ref{cor:G(np)inXnd}, it suffices to show that  \eqref{eq:thrm:GnpAccuracy} holds  
with sufficiently large probability over the instances $\G$, conditional that $\G\in \mathfrak{X}_{n,d}$.

Let ${\cal A}$ be the event $\G\in \mathfrak{X}_{n,d}$.
First we argue about
$\mathbb{E}\left [ ||\mu- \hat{\mu}|| \  | \ {\cal A }\right  ]$, i.e. the expectation is w.r.t. the instances $\G$.
Using Theorem \ref{thrm:Accuracy} and Theorem \ref{thrm:eIsomGnp}
we have that
\begin{displaymath}
\mathbb{E}\left [ ||\mu- \hat{\mu}|| \ | \  {\cal A} \right  ]\leq \mathbb{E}  \left[ \sum_{i=0}^{r-1} \alpha_i  \ |\  {\cal A} \right],
\end{displaymath}
where the expectation is taken over the instances  $\G$.
Noting that  $\alpha_i\in [0,1]$, we get 
\begin{equation}\label{thrm:GnpAccuracyBase}
\mathbb{E} \left [ ||\mu- \mu'||  \ | \  {\cal A} \right ] \leq \displaystyle \sum_{i=0}^{(1+n^{-1/3})dn/2} 
\mathbb{E} [\alpha_i \ |\  {\cal A} ],
\end{equation}
where the above follows by observing that ${\cal A}$ implies that $r\leq (1+n^{-1/3})dn/2$.

On the other hand for the quantities $\mathbb{E} [\alpha_i \ | \ {\cal A} ]$ we work as follows:
\begin{eqnarray}\nonumber
\mathbb{E}[\alpha_i \ | \ {\cal A}]&\leq &\left(\Pr[ {\cal A} ]\right)^{-1}\cdot 
\mathbb{E}[\alpha_i] \qquad\qquad\qquad\qquad \mbox{[since $\alpha_i\geq 0$]}\nonumber \\
&\leq & 100 \epsilon^{-1} k(4+\epsilon)n^{-\left(1+\frac{\epsilon}{36(1+\epsilon/4)\log d} \right)},\label{eq:EAiCondBound}
\end{eqnarray}
in the final inequality we used Theorem \ref{thrm:eIsomGnp} and  Corollary \ref{cor:G(np)inXnd}.
Plugging  \eqref{eq:EAiCondBound} into \eqref{thrm:GnpAccuracyBase},  we get that
\begin{displaymath}
\mathbb{E}\left [ ||\mu- \hat\mu|| \  | \ {\cal A}\right ]\leq C\cdot  n^{-\frac{\epsilon}{36(1+\epsilon/4)\log d}},
\end{displaymath}
for fixed $C>0$. The theorem follows by applying  Markov's  inequality.
\end{theoremproof}

\subsection{Proof of Theorem \ref{thrm:GnpTimeCmplxt}}\label{sec:thrm:GnpTimeCmplxt}

\noindent
First, we are going to show that, on input $G\in \mathfrak{X}_{n,d}$,  
${\tt Random\; Colouring\; Algorithm}$  has time complexity   $O(n^2)$.  Then, the theorem 
will follow by using Corollary \ref{cor:G(np)inXnd}.

We start by considering the time complexity of the algorithm on input  $G\in \mathfrak{X}_{n,d}$. 
First the algorithm  constructs $G_0,\ldots,G_r$. For this, it needs to  distinguish which edges in  
$E(G)$ do not belong to a short cycle.
This can be done by exploring the structure of the $(\log_d n)/9$-neighbourhood 
around each edge of G by  using {\em Breadth First Search}
(BFS).  The search around each edge requires  $O(n)$ steps, since $|E(G)|=O(n)$.
The exploration is repeated for each edge in $E(G)$. Thus,  the algorithm requires 
$O(n^2)$ steps to find the short cycles. This implies that $G_0,\ldots, G_r$ can be 
constructed in $O(n^2)$ steps.

Since the $|E(G_i)|=O(n)$, for every $i=0,\ldots, r$,  Theorem \ref{thrm:UpdateCmplxt} implies that 
the number of steps required for each ${\tt Update}$ call is $O(n)$.
Consequently,   we need $O(n^2)$ steps for 
all the calls of ${\tt Update}$,  since $r\leq |E(G)|=O(n)$.

It remains to consider   the time complexity of  colouring  randomly $G_0$.  
The algorithms uses  ${\tt Random}$ ${\tt  Colouring\; in\;  \mathfrak{S}_{n,d}}$ (Section \ref{sec:AppGnpNew})
to colour randomly  $G_0$.  
Due to our assumptions it holds that $G_0 \in \mathfrak{S}_{n,d}$. 
Let ${\cal C}$ be the set of cycles in $G_0$.   Note that all the cycles in ${\cal C}$ are simple and  
isolated from each other. 
Also, all the vertices in $G_0$ which are not in a cycle are isolated.

We consider the time complexity of colouring the cycles in ${\cal C}$. For  each   $C=(w_0,\ldots, w_{|C|})\in {\cal C}$, 
first,  the problem  is reduced to computing Gibbs marginals  on a tree of maximum degree 2. This is done by assigning 
 $w_0$ a uniformly random colour  from $[k]$. 
Then, the algorithm colours iteratively the vertices in $C$. At iteration $i$, the colouring
of the vertices $w_1,\ldots, w_{i-1}$ is already known and the algorithm colours $w_i$ as follows:
It computes the marginal  $\mu_{w_i}$, conditional the colour assignment of the vertices $w_0,\ldots, w_{i-1}$, 
 by using Dynamic  Programming. Then it assigns a colour to $w_i$ according to $\mu_{w_i}$.

Given the distribution of the children of $w_i$ w.r.t. the subtree that hangs from them,  the Dynamic 
Program requires  $O(k^2)$ arithmetic  operations  to compute $\mu_{w_i}$.  This means that the algorithm
requires $O(k^2  |C|)$ operations for computing $\mu_{w_i}$. It is clear that each cycle $C$ requires 
at most $O(k^2  |C|^2)$ steps  to be coloured randomly.

Consequently,  the algorithm requires  $O(k^2 n\log^2 n)$  number of steps to colour randomly
all the cycles in ${\cal C}$, since $|C|=O(\log n)$ and $|{\cal C}|=O(n)$. Additionally,  the algorithm requires $O(n)$ 
steps to colour randomly all the $O(n)$ many isolated vertices.  

Concluding,   the time complexity of  
${\tt Random\; Colouring\; in \; \mathfrak{S}_{n,d}}$, for  fixed $k$ is $O(n\log^2 n)$.
This implies that  ${\tt Random\; Colouring\; Algorithm}$,  on input  $G\in \mathfrak{E}_{n,d}$,  has  time complexity 
$O(n^2)$.

The theorem follows.

\section{Proof of Theorem \ref{thrm:eIsomGnp}}\label{sec:thrm:eIsomGnp}

Let $\Lambda_{n,k}$ denote the  set of all the 4-tuples $(G,v,u,\sigma)$ such that $G$ is a $k$ colourable  
graph on $n$ vertices, $v,u\in V(G)$ and $\sigma$ is a $k$-colouring of $G$.  For  $(G,v,u,\sigma)\in \Lambda_{n,k}$ 
and   $q\in [k]\backslash\{\sigma_v\}$, consider the disagreement graph $\mathbold{Q}=\mathbold{Q}(G,v,\sigma,q)$
and let  the event  ${\cal Q}_{\sigma_v,q}=\textrm{``$u\in \mathbold{Q}$''}.$

For  $c_1,c_2 \in [k]$   and  an integer $i\geq 0$ we let the distribution   ${\cal P}^i_{c_1,c_2}$  over
$(G,v,u,Z)\in \Lambda_{n,k}$     be induced by the following  experiment:  Take an instance $\G$ and 
construct the sequence $G_0,\ldots, G_r$ as described in Section \ref{sec:AppGnpNew}. Then,
\begin{enumerate}
\item $G$ is equal to  $G_i$
\item $v$ and $u$ are equal to $v_i$ and $u_i$, respectively
\item $Z$ is distributed uniformly at random in $\Omega_{G}(c_1,c_2)$
\end{enumerate}

\begin{remark}
In $G_{0},\ldots, G_r$, the number of terms in the sequence is a random variable.  
In the definition of ${\cal P}^i_{c_1,c_2}$  if $i>r$ we follow the convention that $G$ is the empty graph with 
probability $1$.
\end{remark}

\noindent
Also, denote by ${\cal P}^i_{*,c_2}$  the distribution when $Z(v)$ is not fixed, i.e. $Z$ is a random $k$-colouring
of $G$, conditional that $Z(u)=c_2$. In the same manner, denote by  ${\cal P}^i_{c_1,*}$, the distribution when 
$Z(u)$  is not fixed. Finally, we define ${\cal P}^i_{*,*}$ when there  is no restriction on the colouring of both $v,u$.

For proving Theorem \ref{thrm:eIsomGnp} we need the following two results.
\begin{proposition}\label{prop:DecayOfPaths}
Let $\epsilon$, $d$ and $k$ be as in the statement of Theorem \ref{thrm:eIsomGnp}.
Let $c,q\in [k]$ be   such that $c\neq q$.  For any $i\geq 0$, it holds that
\[
{\cal P}^i_{c,*}[{\cal Q}_{c,q} ] \leq 10\epsilon^{-1}
(4+\epsilon) n^{-\left(1+\frac{\epsilon}{36(1+\epsilon/4)\log d} \right)}.
\]
\end{proposition}
The proof of Proposition \ref{prop:DecayOfPaths} appears in Section \ref{sec:prop:DecayOfPaths}. 
\begin{lemma}\label{lemma:Pont2PontSpatial}
Let $\epsilon$, $d$, $k$ be as in the statement of Theorem \ref{thrm:eIsomGnp}. For any $c\in [k]$
and any $i\geq 0$ it holds that
\[
||{\cal P}^i_{c,*}(\cdot )- {\cal P}^i_{*,*}(\cdot )||_{\{u_i\}}\leq  n^{-1}. 
\]
\end{lemma}
The proof of Lemma \ref{lemma:Pont2PontSpatial} appears in Section \ref{sec:lemma:Pont2PontSpatial}.\\

\begin{theoremproof}{\ref{thrm:eIsomGnp}}
It is elementary to verify that  
\begin{equation}\label{eq:Base4Eai}
\mathbb{E} [\alpha_i]\leq \max_{c,q\in [k]:c\neq q}\left\{ {\cal P}^i_{c,c}[  {\cal Q}_{c,q} ]+{\cal P}^i_{q,c}[ {\cal Q}_{q,c} ]\right\}.
\end{equation}
The theorem  follows by bounding appropriately the probability terms in the r.h.s. of \eqref{eq:Base4Eai}.

Given  $(G,v,u,\sigma)\in \Lambda_{n,k}$, we let the events  $E:=$``$\sigma(v)=\sigma(u)$" and 
$A_{c_1}:=$``$\sigma(u)=c_1$", for every $c_{1}\in [k]$.  Since it holds that  
${\cal P}^i_{c,*}[{\cal Q}_{c,q}] \geq {\cal P}^i_{c,*}[{\cal Q}_{c,q}|E]\cdot  {\cal P}^i_{c,*}[E]$ and 
 ${\cal P}^i_{c,*}[\cdot |E]={\cal P}^i_{c,c}[\cdot ]$, we get that 
\begin{equation}\label{eq:PrBI}
{\cal P}^i_{c,c}[  {\cal Q}_{c,q} ] \leq\frac{1}{{\cal P}^i_{c,*}[E]}  {\cal P}^i_{c,*}[{\cal Q}_{c,q}].
\end{equation}

\noindent
Noting that ${\cal P}^i_{c,*}[E]={\cal P}^i_{c,*}[A_c]$ and  ${\cal P}^i_{*,*}[A_c]=k^{-1}$,
from Lemma \ref{lemma:Pont2PontSpatial} we get that
\begin{equation}\label{eq:PrEi}
\left |{\cal P}^i_{c,*}[E] -k^{-1}\right |\leq n^{-1}.
\end{equation}
Using \eqref{eq:PrEi} and \eqref{eq:PrBI} we get that 
\begin{equation}\label{eq:PiCCUpperBound}
{\cal P}^i_{c,c}[ {\cal Q}_{c,q} ]\leq \displaystyle  2k \cdot {\cal P}^i_{c,*}[ {\cal Q}_{c,q}] 
\leq 20\epsilon^{-1}{k(4+\epsilon)}n^{-\left(1+\frac{\epsilon}{36(1+\epsilon/4)\log d} \right)},
\end{equation}
where the last inequality follows from Proposition \ref{prop:DecayOfPaths}.
Applying the same  arguments we, also, get that
\begin{equation}\label{eq:PiCQUpperBound}
{\cal P}^i_{q,c}[ {\cal Q}_{q,c} ]\leq
20\epsilon^{-1}{k(4+\epsilon)}n^{-\left(1+\frac{\epsilon}{36(1+\epsilon/4)\log d} \right)}.
\end{equation}
The  bounds in \eqref{eq:PiCCUpperBound} and \eqref{eq:PiCQUpperBound} hold for any  
$c,q\in [k]$, different with each other. The theorem follows by plugging \eqref{eq:PiCCUpperBound} 
and \eqref{eq:PiCQUpperBound} into \eqref{eq:Base4Eai}.
\end{theoremproof}

\subsection{Proof of Lemma \ref{lemma:Pont2PontSpatial}}\label{sec:lemma:Pont2PontSpatial}

Let  $(G,v,u,X),(G,v,u,Z) \in \Lambda_{n,k}$, for some fixed $G$. Let   $X, Z$ be   two  coupled random 
colourings of $G$. In particular for  $X,Z$ we have the following: Assuming
that $X(v)=c$,   we   choose  $q$ u.a.r. among $[k]$  and we set $Z(v)=q$.  Depending on whether
$c=q$ or not the coupling does the following choices.
\begin{description}
\item[Case ``$c=q$":]   Couple $Z$ and $X$ identically, i.e. $X=Z$
\item[Case ``$c\neq q$":]  Set $Z={\tt Switching}(G,v,X,q)$, 
\end{description}

\noindent
where ${\tt Switching}$ is  from Section \ref{sec:Problem1}. Claim \ref{claim:OldIsomorph} establishes  
that $Z$ follows the appropriate distribution.
\begin{myclaim}\label{claim:OldIsomorph}
 ${\tt Switching}(G,v, X, q)$ is a random colouring of $G$ conditional 
on that  $v$ is coloured $q$.
\end{myclaim}
\begin{proof}
It suffices to show that the sets $\Omega_c=\cup_{c' \in [k]} \Omega_i(c,c')$ and 
$\Omega_q=\cup_{c'\in[k]}\Omega_i(q,c')$   admit the bijection ${\tt Switching}(G,v,\cdot,q):\Omega_c\to \Omega_q$.

First, note that  Lemma \ref{lemma:SwitchProperColour} implies that if $\tau={\tt Switching}(G,v, \sigma,q)$, 
then  $\tau\in \Omega_{G,k}$. Also,  it is direct that $\tau\in \Omega_{q}$. Second, we need to show that 
the mapping ${\tt Switching}(G,v,\cdot, q):\Omega_c\to \Omega_q$ is  {\em surjective}, i.e. for any 
$\sigma\in \Omega_q$ there is a  $\sigma'\in \Omega_c$ such that $\sigma={\tt Switching}(G,v,\sigma',q)$.
Clearly, such $\sigma'$ exists. In particular, it holds that $\sigma'={\tt Switching}(G,v,\sigma,c)$. The last 
observation also implies that the mapping is {\em one-to-one}. Since ${\tt Switching}(G,v,\cdot,c)$ is surjective
and one-to-one  it is a bijection. The claim follows.
\end{proof}

\noindent
For the case where $q\neq c$,  consider the disagreement  graph $\mathbold{Q}=\mathbold{Q}(G,v, X,q)$. 
We remind the reader that  the event ${\cal Q}_{c,q}:=$``$u\in \mathbold{Q}"$. 
Due to the way we construct  $Z$ we have that  the event ${\cal Q}_{c,q}$ holds if and only if
$X(u)\neq Z(u)$ holds.  
That is,
\begin{equation}\label{eq:lemma:Pont2PontSpatial4FixedG}
\Pr[X(u)\neq Z(u)] \leq  \Pr[{\cal Q}_{c,q}].
\end{equation}
Note that the probability terms above hold for any $k$-colourable graph $G$.

For our purpose, we need to consider $(G,v,u,X),(G,v,u,Z)$ distributed as in ${\cal P}^i_{c,*}$ and
${\cal P}^i_{q,*}$ respectively, for $q\neq c$.  For such 4-tuples,  
 \eqref{eq:lemma:Pont2PontSpatial4FixedG} implies that 
\[
\Pr[X(u)\neq Z(u)]\leq {\cal P}^i_{c,*}[Q_{c,q}].
\]
Note that the above is derived by taking averages w.r.t. the graph  instance $G_i$ in the  sequence 
$G_0,\ldots,G_r$  where  $(v,u)$ correspond to $(v_i,u_i)$. 
The lemma follows by noting that
\[
||{\cal P}^i_{c,*}(\cdot )- {\cal P}^i_{*,*}(\cdot )||_{\{u\}}\leq  {\cal P}^i_{c,*}[Q_{c,q}], 
\]
while from Proposition \ref{prop:DecayOfPaths} we have that ${\cal P}^i_{c,*}[Q_{c,q}]\leq n^{-1}$.

\section{Proof of Proposition \ref{prop:DecayOfPaths}} \label{sec:prop:DecayOfPaths}

Let $(G,v, u,X)$ be distributed as in ${\cal P}^i_{c,*}$.
Every path $P$ in $G$ which start from $v$ and $\forall w\in P$ we have $X(w)\in\{c,q\}$ is called 
 {\em path of disagreement}.
It holds that
\[
{\cal P}^i_{c,*}[{\cal Q}_{c,q} ]\leq 
{\cal P}^i_{c,*}[B ]+{\cal P}^i_{c,*}[C],
\]
where the events $B$ and $C$ are  as follows:  $B:=$``$v$ and $u$ are connected through a path of 
disagreement of length at most $\log^2 n$''.  $C:=$``$v$ and $u$ are connected through a path of length 
greater  than $\log^2 n$'.

Let, also, the event $C':=$ ``there is a path of disagreement starting  from $v$ and has length greater 
than $\log^2 n$". Note that the event $C'$ does not specify the end vertex of the path of disagreement.
It is immediate that ${\cal P}^i_{c,*}[C'] \geq {\cal P}^i_{c,*}[C]$,  since, the event $C$ is included in  the
event $C'$.  Thus, it holds that 
\[
{\cal P}^i_{c,*}[{\cal Q}_{c,q} ]\leq 
{\cal P}^i_{c,*}[B ]+{\cal P}^i_{c,*}[C'].
\]
The proposition will follow by bounding appropriately the probabilities ${\cal P}^i_{c,*}[B]$ and ${\cal P}^i_{c,*}[C']$.

For every vertex $w$, we  let $\Gamma_{w}(l)$ denote the number of paths  of disagreement of length $l$ that 
connect $v$ and $w$. From Markov's  inequality  we get that
\begin{equation}\label{eq:calPBBound}
{\cal P}^i_{c,*}[B]\leq \mathbb{E}_{{\cal P}^i_{c,*}}\left[\sum_{l\leq \log^2 n}\Gamma_{u}(l)\right],
\end{equation}
where $\mathbb{E}_{{\cal P}^i_{c,*}}[\cdot]$ is the expectation w.r.t. ${\cal P}_{i_{c,*}}$.
For bounding ${\cal P}^i_{c,*}[C']$ we use the following inequality
\begin{equation}\label{eq:calPCBound}
{\cal P}^i_{c,*}[C']\leq \mathbb{E}_{{\cal P}^i_{c,*}}\left[\sum_{w}\Gamma_w(\log^2 n)\right],
\end{equation}
where the summation on the r.h.s. of the inequality, above, runs over all the vertices of the graph.

So as to compute the expectation both in \eqref{eq:calPBBound} and \eqref{eq:calPCBound}
we use Theorem \ref{thrm:PrIPGnp}. However,  we note  that the pair of vertices $v,u$ we consider
is not a uniformly random one. Since  we consider the probability distribution ${\cal P}^i_{c,*}$, the pair $v,u$
is distributed uniformly at random among the pair of vertices which are at distance greater than $(\log_dn)/9$
in $\G$.

Letting   $p$ be the probability that  a randomly chosen edge from $\G$ does not belong to a cycle of length 
less than $(\log_d n)/9$.  Using Theorem \ref{thrm:PrIPGnp}  we get that
\begin{equation}
\mathbb{E}_{{\cal P}^{i}_{c,*}} \left[\sum_{l\leq \log^2 n}\Gamma_{u}(l)\right] 
\leq 2p^{-1}\sum^{\log^2 n}_{l\geq l_0} n^{l-1}\left({(1+\epsilon/4)n}\right)^{-l},
\qquad \textrm{for }l_0=(\log_d n)/9+1.
\end{equation}
Let us explain how do we get the above inequality from {Theorem
 \ref{thrm:PrIPGnp}}. If the vertices $v,u$ were not conditioned to be at distance 
 greater than $(\log_d n)/9$, then the expected number of paths of disagreement of length $l$ between
them   is equal to the number of possible paths of length
$l$ times the probability each of these paths is a path of disagreement. Clearly the number of the possible
paths is at most  $n^{l-1}$, i.e. we have fixed the first and the last vertex of the paths. From 
{Theorem   \ref{thrm:PrIPGnp}}
we have that the probability of each of these paths to be disagreeing is $2\left({(1+\epsilon/4)n}\right)^{-l}$.
We divide by $p$ due to conditioning that the vertices $v,u$  are not entirely  random, since we 
have conditioned that their distance is larger than $(\log_d n)/9$.

It is direct to show that it holds that $p\geq 1-n^{-9/10}$. Then, we have that  
\begin{equation}\label{eq:FinalBoundBPathDis}
\mathbb{E}_{{\cal P}^{i}_{c,*}} \left[\sum_{l\leq \log^2 n}\Gamma_{u}(l)\right] 
\leq 4\epsilon^{-1} {(4+\epsilon)}n^{-1-\frac{\epsilon}{36(1+\epsilon/4)\log d}}.
\end{equation}
Working in the same manner for  \eqref{eq:calPCBound} we get that
\begin{eqnarray}
\mathbb{E}_{{\cal P}^{i}_{c,*}}
\left[\sum_{w}\Gamma_w(\log^2 n)\right]
&\leq& 2p^{-1}\left(1+\epsilon/4 \right)^{-\log^2 n} \nonumber \\
& \leq&2 p^{-1}n^{-((\log n) \cdot\log(1+\epsilon/4))} \leq n^{-\sqrt{\log n}},\label{eq:FinalBoundCPathDis}
\end{eqnarray}
where the last inequality holds for  large $n$ and noting that $p>1/2$. Observe that in the second case the 
number of paths of length $l$ that emanate from $v$ is  at most $n^l$, as  we do not fix the last vertex of the path.

Using \eqref{eq:FinalBoundBPathDis}  and \eqref{eq:calPBBound} we bound 
appropriately ${\cal P}^i_{c,*}[B]$. Using  \eqref{eq:FinalBoundCPathDis}
and  \eqref{eq:calPCBound} we bound appropriately ${\cal P}^i_{c,*}[C']$. The proposition follows.

\section{Proof of Theorem \ref{thrm:PrIPGnp}}\label{sec:thrm:PrIPGnp}

For the sake of brevity we denote with $P$ not only the permutation of the vertices $w_0,\ldots,w_{\ell}$
but the corresponding path in $\SG$, if such path exists.
The probability term in \eqref{eq:thrm:PrIPGnp} is w.r.t. both the randomness of the graph $\SG$  and the random $k$-colourings
of $\SG$.  That is, for  $\mathbf{I}_{\{P\}}=1$, first we need  to have that the  vertices in the permutation 
$P$  form  path in $\SG$.  Then,  given that  $\SG$ contains the path $P$, we need to bound the probability
that this path is 2-coloured in a random $k$-colouring of $\SG$.  Clearly, the challenging part is the second one.
We denote $\SG_P$ the  graph  $\SG$ conditional that the path  $P$ appears in the graph.

Our approach  is as follows: Given $\SG_P$,  first we specify an appropriate subgraph of $\SG_P$ which 
includes the path $P$. We call this subgraph $\mathbold{N}$. Also, we specify a set $\mathbold{B}\subset V(\mathbold{N})$ 
such that $\mathbold{B}$ separates $V(\mathbold{N})\backslash \mathbold{B}$ from the rest of the graph $\SG_P$.
We   set  an appropriate  (worst case) boun\-dary condition $\sigma_{\mathbold{B}} \in [k]^{\mathbold{B}}$
on $\mathbold{B}$. 
Let $\mu^{\sigma}_{N}$, be the Gibbs distribution  of the $k$-colourings of $\mathbold{N}$, conditional  that 
$\mathbold{B}$ is coloured $\sigma_{B}$.
The choice of $\sigma$ is such that under $\mu^{\sigma}_{N}$ the  probability of $P$ to be 2-coloured with 
$c,q$ is lower bounded by the corresponding probability under  $\mu_H$,  the Gibbs distribution of the $k$-colourings
of $\SG_P$.

Let us describe how do we get $\mathbold{N}$ and $\mathbold{B}\subset V(\mathbold{N})$.
For this, we consider  an integer parameter $h=h(\epsilon)>0$, which we assume that is  sufficiently large
it depends on $\epsilon$ and it is  independent of $d$.
\\ \vspace{-.1cm}

\noindent
{\bf Path Neighbourhood Revealing.}
Consider the graph $\SG_P$. For each  $w_i\in P$ we  define  the sets ${ L}_{i,s} \subseteq V(\SG_P)$, 
for $s=0,\ldots, h$, as follows:
 ${ L}_{i,0}=\{w_i\}$.
We get  ${ L}_{i,s}$ by working inductively, i.e. we use  ${ L}_{i,s-1}$. 
Let ${\cal R}_{i,s} \subset V(\G)$ contain all the vertices  but those 
which belong to  $P$
and those which belong in $\bigcup_{j <  i}\bigcup_{j'\leq h}{ L}_{j,j'} $ and $\bigcup_{j'<s}{ L}_{i,j'}$. 
Consider an (arbitrary) ordering of the vertices in ${\cal R}_{i,s}$. For each vertex  $u\in { L}_{i,s-1}$ we 
examine its adjacency with the  vertices in ${\cal R}_{i,s}$ in the   predefined order. We stop revealing 
the neighborhood of $u$ in ${\cal R}_{i,s}$ once we either have revealed $(1+\epsilon/3)d+1$ many neighbours,  
or if  we have checked all the possible adjacencies of $u$  with ${\cal R}_{i,s}$. Whichever happens 
first\footnote{\label{foot:DegLv} Clearly, as the process goes,  the number of neighbours of $u$ in ${\cal R}_{i,s}$ is at most
$(1+\epsilon/3)d$+1.}. Then $L_{i,s}$ contains all the vertices in ${\cal R}_{i,s}$ which have been revealed
to have a neighbour in $L_{i,s-1}$.

For $i=0,\ldots, \ell$,  let   ${ N}_{i,h}$ be the induced subgraph of $\SG_P$ with vertex set 
 $\bigcup^h_{s=0}{ L}_{i,s}$.  Note that the size of ${ N}_{i,h}$ depends only on $\epsilon, d, h$, i.e. 
 it is independent of $n$. In particular,    it holds that   
\begin{equation}\label{eq:SizeOfNih}
\left |V(N_{i,h}) \right |\leq  N_0=\frac{[(1+\epsilon/3)d+1]^{h+1}-1}{(1+\epsilon/3)d}.
\end{equation}

\noindent
We call  ${ N}_{i,h}$,  {\tt Fail} if at least one of the following happens:
\begin{itemize}
 \item The maximum degree in ${N}_{i,h}$ is at least  $(1+\epsilon/3)d+1$
\item The graph ${N}_{i,h}$ is not a tree
\item There is an  integer $j\neq i$ such that some vertex $w''\in {N}_{j,h}$
is adjacent to some vertex $w'\in {\cal N}_{i,h}$ and the edge $\{w',w''\}$ does not belong to 
 the path  $P$.
\end{itemize}
\begin{lemma}\label{lemma:failProb}
Let  $\epsilon,d$ be  as in Theorem \ref{thrm:PrIPGnp}. 
Consider a sufficiently large fixed integer $h=h(\epsilon)>0$, independent of $d$.  Let $F$ be the number of 
vertices $w_i\in P$ such that $N_{i,h}$ is ${\tt Fail}$, for $i=1,\ldots, \ell$. For any $s=1,\ldots, \ell$, it holds that 
\begin{equation}
\Pr[F=s]\leq (1+n^{-1/3}){\ell \choose s} \exp\left[-\epsilon^2{d}s/{35} \right]. \nonumber
\end{equation}
\end{lemma}
In the lemma, above,  $F$ does not consider $N_{0,h}$.
The proof of Lemma \ref{lemma:failProb} appears in Section  \ref{sec:lemma:failProb}.\\

\noindent
The graph  $\mathbold{N}$ we are looking for is a subgraph of $\bigcup^{\ell}_{i=0}{ N}_{i,h}$.
For specifying  $\mathbold{N}$ perhaps it is more natural to  start with  the set $\mathbold{B}$ which separates
$\mathbold{N}$ from the rest of $\mathbold{H}_P$.
Each time, we  decide on $\mathbold{B}\cap V({ N}_{i,h})$ by  examining 
each ${ N}_{i,h}$, separately. 
If ${ N}_{i,h}$ is {\tt Fail}, then  $\mathbold{B}\cap V({ N}_{i,h})=\{w_i\}$, i.e the
vertex in the path $P$.
On the other hand,  if   ${ N}_{i,h}$ is not {\tt Fail}, then 
$\mathbold{B}\cap V({ N}_{i,h})={ L}_{i,h}$, i.e.  all the vertices in ${ N}_{i,h}$
that are at distance $h$ from $w_i$.

\begin{figure}
	\centering
		\includegraphics[height=2.3cm]{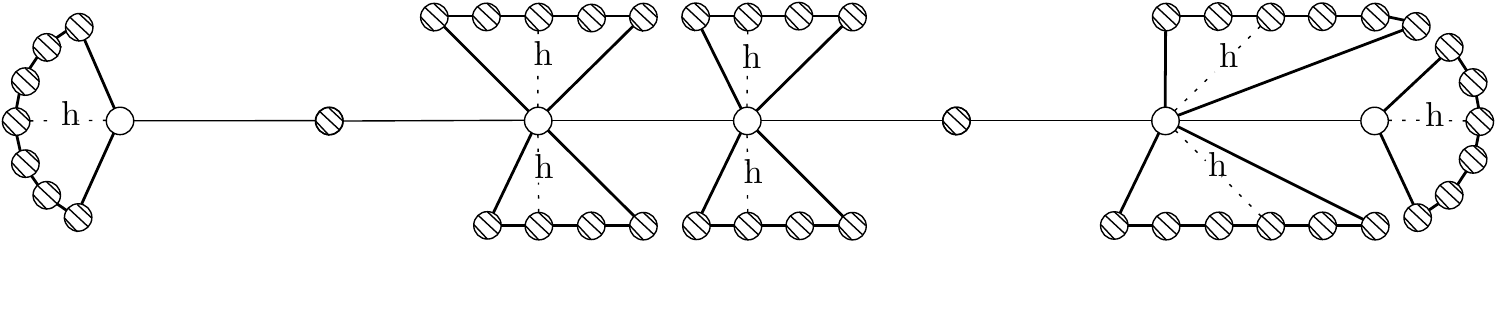}
		\caption{The lined vertices belong to $\mathbold{B}$.}
	\label{fig:NewApprRefined}
\end{figure}
In Figure \ref{fig:NewApprRefined}, we see one example of a possible outcome of the
exploration we describe above.  The lined vertices are exactly those which belong to the boundary
set $\mathbold{B}$. If some vertex $w_i$ on the path is lined, this means that $N_{i,h}$  is 
{\tt Fail}. The vertices of the path which are not lined correspond to the roots of a ``low degree"
tree of height at most $h$.

Let $S\subseteq\{0,\ldots, \ell\}$ contain each $i$ such that $N_{i,h}$ is not {\tt Fail.}
Also,  let $V_A=\bigcup_{i\in S}V(N_{i,h-1})$  \footnote{ $N_{i,h-1}$ is defined in the natural way}.  It is not hard to see that the vertex set
 $\mathbold{B}$ is a cut-set that separates $V_A$ from the rest of the vertices
in $V(\SG_P)$. The graph $\mathbold{N}$ is the induced subgraph of $\SG_P$ with vertex
set $V_A\cup \mathbold{B}$. 
\begin{remark}
Since $\SG_P$ is random, the subgraph $\mathbold{N}$ is random.
\end{remark}

\noindent
Consider  the graph $\SG_P$ and the corresponding Gibbs distribution $\mu_H$.
The distribution $\mu_{H}$ specifies a convex combination of boundary conditions
on $\mathbold{B}$. Using these boundary conditions we could estimate the
probability that $P$ is coloured only with $c,q$, exactly. 
However, estimating this convex combination of boundaries is a formidable task to accomplish. 
We get an upper bound of this probability by considering a worst boundary condition
on the vertex set $\mathbold{B}$. The condition is worst in the sense that it maximizes
the probability of interest. That is, instead of $\mu_H$, we consider the distribution
$\mu^{\sigma}_N$ which is much easier to handle.
Under  $\mu^{\sigma}_N$ the probability that $P$ is coloured with $c,q$  is at least as  big as
under $\mu_H$.

In the following results, we let ${\cal T}_{d,\epsilon,h}$ be the set of labeled, rooted, trees of 
maximum degree $(1+\epsilon/3)d$ and height $h$.
\begin{proposition}\label{prop:spatial-mixing-trees}
Let $\epsilon, d, k$ be as in Theorem  \ref{thrm:PrIPGnp}.  
Consider a sufficiently large fixed integer $h=h(\epsilon)>0$, independent of $d$.
Consider $\SG_P$ and
let $\mathbold{N, B}$ be as  defined above.
For each   $w_j\in P$  such that $w_j \notin \mathbold{B}$ the following is true:

Let $\Gamma$ be the neighbours of $w_j$ in the path $P$ and let $\mathbold{B}^+=\mathbold{B}\cup \Gamma$.
 There exists a  function $f_{\epsilon}:\mathbb{N}\to\mathbb{R}^+$, such that 
 $f(h)\to 0$ as $h\to \infty$, while
for any $\sigma\in \Omega_{\mathbold{N},k}$  and any  $c\in [k]$  it holds that
\begin{equation}
\max_{{N}_{j,h}\in {\cal T}_{d, \epsilon, h}}\left|
\Pr[X(v_j)=c \ | \ {N}_{j,h},  X_{\mathbold{B}^+}=\sigma_{\mathbold{B}^+}]-
\Pr[X(v_j)=c\ |\ {N}_{j,h},  X_{\Gamma}=\sigma_{\Gamma} ]
\right |\leq  k^{-1} {f_{\epsilon}(h)}, \nonumber
\end{equation}
where $X$ is  a random $k$-colouring of $\mathbold{N}$. 
\end{proposition}
Note that the above  is a spatial mixing result. It implies that for any $N_{j,h}$ which is not {\tt Fail} the boundary we 
set at distance $h$ from $w_j$, essentially, has no effect on the distribution of the $k$-colouring of $w_i$.
The proof of Proposition \ref{prop:spatial-mixing-trees} appears in Section \ref{sec:prop:spatial-mixing-trees}.

For every  $w_j\in P$ such that $w_j\in \mathbold{B}$, the worst case boundary condition 
sets the vertex to its appropriate colour,  i.e. if $j$ is even then the colour is  $c$, otherwise the colour is  $q$.
Proposition \ref{prop:spatial-mixing-trees} implies that, whatever is the boundary condition at $\mathbold{B}$,
if $w_j\notin \mathbold{B}$, its probability of getting  colour $q$ or $c$, depending on
the parity of $j$,  is approximately  $1/k$.
\\

\begin{theoremproof}{\ref{thrm:PrIPGnp}}
Let  $E_{P}$ be the event that $\mathbold{H}$ contains the path $P$. It holds that 
\begin{displaymath}
\Pr\left[\mathbf{I}_{\{P\}}=1 \right] \leq   \left (d/n\right)^{\ell} \cdot \Pr\left[\mathbf{I}_{\{P\}}=1 \ | \ E_{P}\right].
\end{displaymath}
Consider  $\SG_P$ and let  $X$ be  a random $k$-colouring  conditional on that $X(w_0)=c$.  
For $i$ even, we call $w_i \in P$ disagreeing if $X(w_i)=c$. For $i$ odd number,  we call 
$w_i \in P$ disagreeing if $X(w_i)=q$.

Let the event $D_i$ that ``$w_i$ is disagreeing". 
Clearly it holds that
\begin{equation}\label{eq:PrIpCondEpBoundBase}
\Pr\left[ \mathbf{I}_{\{P\}}=1   \right]\leq \left (d/n\right)^{\ell} \Pr \left[ \cap^{\ell} _{i=1} D_i \ |\  E_P\right].
\end{equation}

\noindent
Let the events $A_i, B_i, C_i$ be defined as follows:
$A_i=$ ``${N}_{i,h}$ is {\tt Fail}". $B_i=$ ``${N}_{i,h}$ is not {\tt Fail} and $w_i$ is disagreeing". 
Also let  $C_i=A_i \cup B_i$.
\begin{myclaim}\label{claim:ReductionProbDis}
It holds that 
$$\Pr \left[ \cap^{\ell} _{i=1} D_i \ |\  E_{P}\right] \leq \Pr \left[ \cap^{\ell} _{i=1} C_i\ |\ E_P\right] .$$
\end{myclaim}
\begin{proof}
In the setting of the proof of Theorem \ref{thrm:PrIPGnp},   assume that we have revealed the 
underlying graph $\SG_P$.  It suffices to show that 
\begin{equation}\label{eq:FromD_i2CiDisagreement}
\Pr \left[  \left.  \cap^{\ell} _{i=1} D_i \ \right  |\   \SG_P  \right] 
\leq \Pr \left[ \left . \cap^{\ell} _{i=1} C_i \ \right | \  \SG_P  \right].
\end{equation}
Observe that the probability terms are only w.r.t. the random colouring of $\SG_P$.

Let $W$ be the set of vertices $q_i\in P$ such that ${ N}_{i,h} $ is not {\tt Fail}. Also, let 
$W'\subseteq \mathbold{B}$ be the set of vertices $w_i\in P$ for which  ${ N}_{i,h}$ is {\tt Fail}. 
The events $\cap_{w_i\in W} C_i$ and $\cap_{w_i\in W} D_i$  are identical, since both occur 
if the vertices in $W$ are disagreeing. Thus it holds that 
$\Pr \left[ \left . \cap_{w_i\in W} D_i\ \right | \  \SG_P\right] = \Pr \left[ \left . \cap_{w_i\in W} C_i \ \right |\ \SG_P\right]$.

Furthermore, we note that
$\Pr \left[ \left . \cap_{w_i\in W'} C_i \ \right | \   \SG_P, \cap_{w_i\in W} C_i \right]=1$.
On the other hand, it holds that  $\Pr \left[ \left. \cap_{w_i\in W'} D_i  \  \right| \ \SG_P,  \cap_{w_i\in W} D_i \right]\leq 1$.
These imply that   \eqref{eq:FromD_i2CiDisagreement} is true.  The claim follows. 
\end{proof}

\noindent
Using   Claim \ref{claim:ReductionProbDis} and \eqref{eq:PrIpCondEpBoundBase},   it suffices to  bound appropriately 
$\Pr \left[ \cap^{\ell} _{i=1} C_i \ | \ E_P\right] $.

Consider $\SG_P$ and  let ${\cal F}_i(C)$ be the $\sigma$-algebra generated by the events
 $C_j$, for every $j\neq i$. 
Proposition   \ref{prop:spatial-mixing-trees} implies that 
\begin{equation}\label{eq:ConseqFromSMTrees}
\rho=\Pr[B_i \ | \ {\cal F}_i(C), E_P, N_{i,h} \textrm{ is not }  {\tt Fail}] \leq (k-2)^{-1}+{f_{\epsilon}(h)}/{k}.
\end{equation}
for any $i=0,\ldots, \ell$. 
Letting  $F$ be the number of vertices $w_i\in P$ such that $N_{i,h}$ is {\tt Fail}, for
$i=1,\ldots, \ell$,  we have that 
\begin{eqnarray}
\Pr \left[ \cap^{\ell} _{i=1} C_i \ | \ E_P\right] &=&
 \sum^{\ell}_{s=0} \Pr \left[ \cap^{\ell} _{i=1} C_i \ | \  E_P, F=s\right] \Pr \left[ F=s \ |\ E_P\right]\nonumber\\
&\leq &  \sum^{\ell}_{s=0} \rho^{\ell-s} \Pr \left[ F=s \ |\ E_P\right] \hspace{3.5cm} \mbox{[from \eqref{eq:ConseqFromSMTrees}]} \nonumber\\
&\leq &  (1+n^{-1/3})\sum^{\ell}_{s=0} {\ell  \choose s} \rho^{\ell-s} \exp(-\epsilon^2 ds/35)
\qquad \mbox{[from  Lemma  \ref{lemma:failProb}]} \nonumber\\
&\leq & 2\left [ \rho+\exp(-\epsilon^2 d/35) \right] ^{\ell}
\label{eq:PrCjBoundA}.
\end{eqnarray}
Using the fact that $k\geq (1+\epsilon)d$, for sufficiently large $h, d$, 
 \eqref{eq:PrCjBoundA} implies that
\begin{equation}
\Pr \left[ \cap^{\ell}_{i=1} C_i \ | \  E_P\right] \leq 
2((1+\epsilon/4)d)^{-\ell}. \label{eq:PrCjBoundB}
\end{equation}
The theorem follows from   \eqref{eq:PrCjBoundB}, \eqref{eq:PrIpCondEpBoundBase} and
Claim \ref{claim:ReductionProbDis}.
\end{theoremproof}

\subsection{Proof of Lemma \ref{lemma:failProb}}\label{sec:lemma:failProb}
For proving the lemma we  use the following tail bound,  \cite{janson}, Corollary 2.3.
Let $W$ be distributed as in ${\cal B}(n,d/n)$,  i.e. binomial distribution with parameters $n$ and $d/n$.
For any fixed  $\alpha>0$ and sufficiently large $d$, it holds that
\begin{equation}\label{eq:JansChern}
\Pr[W\geq (1+\alpha)d]\leq \exp\left( -\alpha^2d/3\right).
\end{equation}

\noindent
For $i,j=0,\ldots, \ell$ consider the following events:   Let 
$A_i:=$`` ${N}_{i,h}$ has  maximum degree greater than  $(1+\epsilon/3)d$".   
Also, let   $B_i:=$``${N}_{i,h}$ is not a tree".   For any two $i,j$ such that $i\neq j$, we let  $E_{i,j}:=$``there is 
an edge, not in $P$, which connects some vertex in $N_{i,h}$ and some vertex in $N_{j,h}$".

Given some $i\in\{0,\ldots, \ell\}$ and any  $S\subset \{0,\ldots, \ell\}$ such that $i\notin S$, 
let ${\cal F}_{S}$ be the $\sigma$-algebra generated be the events $A_j,B_j$ for $j\in S$. 
Given, ${\cal F}_S$, for every vertex $w\in { L}_{i,t-1}$ has a  number of  neighbours in ${\cal R}_{i,t}$ 
which is 
dominated by ${\cal B}(n,d/n)$, for  $t=1,\ldots, h$. Then, \eqref{eq:JansChern} implies
that  the probability for $w$ to have at least $(1+\epsilon/3)d$ neighbours in ${\cal R}_{i,t}$ is at 
most $\exp\left( -\epsilon^2 d/27\right).$ 

The event $A_i$ occurs if there exists  $t\in [h]$ and $w\in { L}_{i,t-1}$ whose number of neighbour in 
${\cal R}_{i,t}$ is at least $(1+\epsilon/3)d$. 
A simple union bound over the vertices in $N_{i,h}$ implies the following: for every 
$i=0,\ldots, \ell$ we have that
\begin{equation}
\Pr\left[ A_i \ \left |  \  {\cal F}_{S} \right. \right ]  \leq  N_0\exp\left ( -\epsilon^2d/27\right) \leq  \exp\left ( -\epsilon^2d/30\right),
 \label{eq:RedEvent1}
\end{equation}
where $N_0$ is defined in \eqref{eq:SizeOfNih}.
Also,  it holds that  
\begin{equation}
\Pr\left[ B_i \ \left | \  {\cal F}_S  \right. \right ]  \leq   {N_0 \choose 2} \frac{d}{n}\leq \frac{d^{5h}}{n} .
\label{eq:RedEvent2}
\end{equation}
The above follows by noting $B_i$ occurs, if there is an edge between the vertices $N_{i,h}$ which
is not exposed during the revelation of the sets $\bigcup^h_{s=0}{ L}_{i,s}$. The probability of having
such an edge is upper bounded by the expected number of such edges.

Combining  \eqref{eq:RedEvent1} and \eqref{eq:RedEvent2} with  a simple union bound we get that
\begin{equation}
\Pr\left[ A_i\cup B_i \ \left | \  {\cal F}_{S} \right. \right ] \leq \exp\left( -\epsilon^2d/35\right).
\label{eq:UnionEvent12}
\end{equation}
Let $R$ be the number of subgraphs $N_{i,h}$, for $i\in \{1,\ldots, \ell\}$,  such that  the event $A_i\cup B_i$ holds.  
Eq. \eqref{eq:UnionEvent12} implies that for $R$ we have the following:
For any $x\in \{1,\ldots, \ell\}$ it holds that
\begin{equation}
\Pr[R=x] \leq {\ell  \choose x}z^x_0(1-z_0)^{\ell-x}, 
\end{equation}
where $z_0=\exp\left( -\epsilon^2d/35\right).$
Also, we have that
 \begin{eqnarray}
 \Pr[F=s] &=&\sum^s_{x=0}\Pr[R=x]\Pr[F=s \ | \ R=x] \nonumber \\
 &\leq &\sum^s_{x=0} 
 {\ell \choose x}z^x_0(1-z_0)^{\ell-x}  \Pr[F=s \ | \  R=x] \nonumber \\
&\leq &  \sum^s_{x=0}  {\ell \choose x}z^x_0 \cdot \Pr[F=s\  | \ R=x],  \label{eq:Basis4failProb}
 \end{eqnarray}
 where the last inequality follows from the fact that $(1-z_0)^{\ell-x}\leq 1$.
 
We proceed by bounding appropriately the quantity $\Pr[F=s \ | \ R=x]$.
For this, let $Z$ be the number of pairs of subgraphs $N_{i,h}, N_{j,h}$  for which the event 
$E_{i,j}$ holds, for $i,j=0,1,\ldots, \ell$.
Given that $R=x$,  so as to have $F=s$ there  should be   at least $\lceil  (s-x)/2\rceil $ pairs 
$N_{i,h}, N_{j,h}$ such that $E_{i,j}$ holds, i.e. 
\begin{equation}
\Pr[F=s \ | \  R=x]\leq  \Pr[Z\geq \lceil (s-x)/2\rceil  \ | \ R=x].  \label{eq:FVsZFailProb}
\end{equation}
Given   some $i$ and $j$,   let $J_1$ be a subset of  events  $E_{i',j'}$ such that $E_{i,j}\notin J_1$.
Also, let $J_2$ any subset of events $A_{i'},B_{i'}$. Let   ${\cal F}_{J}$ be the $\sigma$-algebra generated by the
events in $J_1\cup J_2$. 

Noting that the expected number of edges between $N_{i,h}$ and $N_{j,h}$ is at most 
$N^2_0d/n$,  we have that 
\[
\Pr\left[ E_{ij}  \ | \ {\cal F}_{J} \right ]  \leq  N^2_0d/n \leq d^{5h}/n.
\] 
 The above  inequality implies that
for any integer $x\geq 0$ and  $z_1=d^{5h}/n$, we have
\begin{eqnarray}%
\Pr[Z\geq x] &\leq &\sum_{r\geq x}{{\ell +1 \choose 2} \choose r}(z_1)^r(1-z_1)^{{\ell+1 \choose 2}-r} \nonumber \\
&\leq &\sum_{r\geq x}  {{\ell+1 \choose 2} \choose r}(z_1)^r \;\leq \; \sum_{r\geq x} \left( \frac{(\ell+1)^2ez_1}{2r} \right)^r
\nonumber \qquad \mbox{[since ${n\choose i}\leq (ne/i)^i$]}\\
&\leq & 2\left( \frac{(\ell+1)^2ez_1}{2x} \right)^x \; \leq \; (4n^{-1} \log ^4 n)^x, \label{eq:BinomailZ}
\end{eqnarray}
where the last inequality follows due to our assumption that $\ell\leq (\log n)^2$.

Plugging \eqref{eq:BinomailZ} , \eqref{eq:FVsZFailProb} into \eqref{eq:Basis4failProb} we get that
\begin{eqnarray}
 \Pr[F=s] &\leq & \sum^s_{x=0} 
 {\ell \choose x}z^x_0 (4n^{-1} \log ^4 n)^{(s-x)/2} \nonumber \\
%%%%
 &\leq & \sum^s_{x=0} 
 {\ell \choose s-x }z^{s-x}_0 (2n^{-1/2} \log ^2 n)^{x} \nonumber \\
 %%%%
 &\leq & {\ell \choose s }z^s_0 \  \sum^s_{x=0} 
 {\ell \choose s-x }{\ell \choose s}^{-1} [(2/z_0)n^{-1/2} \log ^2 n]^{x} \nonumber \\
 %%%%%
 &\leq & {\ell  \choose s }z^s_0 \  \sum^s_{x=0} 
 \frac{s!}{(s-x)!}\frac{(\ell-s)!}{(\ell-s+x)!}
 [(2/z_0)n^{-1/2} \log ^2 n]^{x} \nonumber \\
 %%%%%
 &\leq & {\ell  \choose s }z^s_0 \  \sum^s_{x=0} 
\left(\frac{s}{\ell-s+1} \right)^x
 [(2/z_0)n^{-1/2} \log ^2 n]^{x} \nonumber \\
 %%%%%
 &\leq & {\ell  \choose s }z^s_0 \  \frac{1}{1-n^{-2/5}}, \nonumber
\end{eqnarray}
where in the last inequality we use the fact that $s\leq \ell\leq (\log n)^2$ and $z_0=\Theta(1)$.
The lemma follows.

\section{Proof of Proposition \ref{prop:spatial-mixing-trees}}\label{sec:prop:spatial-mixing-trees}

For some vertex $w_j\in P$ such that $w_j\notin \mathbold{ B}$ we have that  ${N}_{j,h}$ is  not {\tt Fail}.
That is,  $N_{j,h}$   is a tree of  maximum degree less than $(1+\epsilon/3)d$.
For such $N_{j,h}$ we assume   $w_j$  to be the root.

If the  height of $N_{j,h}$ is less than $h$, then  no vertex in ${ N}_{j,h}$ 
belongs to $\mathbold{B}$.  For such  tree,   the proposition is trivially true.
For the rest of the proof we assume that  the height of ${N}_{j,h}$ is $h$.

From \cite{TreeUniq} we have the following theorem.
%, 
\begin{theorem}[Jonasson 2001]\label{thrm:TreeUniq}
Let $\Delta, h$ be sufficiently large integers and let $k\geq \Delta+2$.  Let $T$  be  a 
complete $\Delta$-ary  tree of height $h$. Let $r$ be the root and let $L$ be  the leaves of $T$.
Also,  let $X$ be a random $k$-colouring of the tree.  For any $c\in [k]$ it holds that
\[
\max_{\sigma\in \Omega_{T,k}}\left|\Pr[X(r)=c \ |\  X(L)=\sigma_L]-{k}^{-1}\right|\leq k^{-1} {\phi_{k}(h)},
\]
where the quantity $\phi_{k}(h) \geq 0$ which tends to zero as $h\to\infty$.
\end{theorem}

\noindent
Theorem \ref{thrm:TreeUniq} establishes the {\em Gibbs uniqueness} condition for the random colourings 
of a $\Delta$-ary  tree. In  Proposition \ref{prop:spatial-monotone} we extend the previous result to  trees of 
{\em maximum degree} $\Delta$.
%i
\begin{proposition}\label{prop:spatial-monotone}
Let $\Delta, h$ be  sufficiently large integers and $k\geq \Delta+2$.  Let $T$ be a 
tree of height $h$ and {\em maximum} degree at most $\Delta$.
Let $r$, $L_0$ denote the root and the vertices at level $h$, respectively.
For $X$ a random $k$-colouring of $T$, the following is true:

For $\phi_{k}(h)$ as in Theorem \ref{thrm:TreeUniq} and for any $c\in [k]$ it holds
that 
\[
\max_{\sigma\in \Omega_{T,k}}\left|\Pr[X(r)=c \ | \ X(L_0)=\sigma_{L_0}]-{k}^{-1}\right|\leq k^{-1} \phi_{k}(h).
\]
\end{proposition}
 The proof of Proposition \ref{prop:spatial-monotone} appears in Section \ref{sec:prop:spatial-monotone}.\\

\begin{propositionproof}{\ref{prop:spatial-mixing-trees}}
We let $\mu_{N}$ be the Gibbs distribution over the $k$-colourings of $\mathbold{N}$, while 
we  let $\mu_{w_j}$ be the marginal of $\mu_{N}$ on $w_j\in P$.
For   $\sigma\in \Omega_{\mathbold{N},k}$ we  let  $t_{\sigma}\subseteq [k]$  contain all the colours 
 that are used from $\sigma$ to colour the vertices in $\Gamma$. It is elementary that $|t_{\sigma}|\leq 2$.
Also,   it holds that
\begin{equation}\label{eq:PrXGammaBound}
\Pr[X(v_j)=c \ |\ {N}_{j,h},  X_{\Gamma}=\sigma_{\Gamma} ]=(k-|t_{\sigma}|)^{-1},
\end{equation}
since we have assumed that ${N}_{j,h}$ is not {\tt Fail}, the structure of $N_{j,h}$ is treelike. 
The above  holds for any   ${N}_{j,h}\in  {\cal T}(d,\epsilon, h)$.

Let  $\mathbold{N}'$ be the graph derived from $\mathbold{N}$ be deleting the edges of
$P$ which are incident to $w_j$.
Let $\nu$ be the Gibbs distribution over  the $k$-colourings of $\mathbold{N}'$,
while  let $\nu_{w_j}$ be the marginal of $\nu$ on $w_j$. 
For any $\sigma\in\Omega_{N,k}$ and  any $c\in [k]\backslash t_{\sigma}$,
let $X$ be a random $k$-colouring of $\mathbold{N}$, then
\begin{equation}\label{eq:MuSigmaVsNuSigma}
\Pr[X(v_j)=c \ |\ {N}_{j,h},  X_{\mathbold{B}^+}=\sigma_{\mathbold{B}^+}]=
\frac{\nu^{\sigma_{\mathbold{B}^+}}_{w_j}(c)}{1-\nu^{\sigma_{\mathbold{B}^+}}_{w_j}\left( t_{\sigma} \right)},
\end{equation}
where $\nu^{\sigma_{\mathbold{B}^+}}_{j}(\cdot)$ denotes the distribution
$\nu_{j}$ conditional that $\mathbold{B}^+$ is coloured $\sigma_{\mathbold{B}^+}$.

The proposition will follows by showing that   the r.h.s. of \eqref{eq:MuSigmaVsNuSigma}
and \eqref{eq:PrXGammaBound} are sufficiently close.
For this,  we need to estimate $\nu^{\sigma_{\mathbold{B}^+}}_{w_j}(c)$.
In particular, we show that for  any $c\in [k]$ it holds that
\begin{equation}\label{eq:target-mu'}
\left|\nu^{\sigma_{\mathbold{B}^+}}_{w_j}(c)-k^{-1}\right|\leq k^{-1}\cdot \phi_{k}(h),
\end{equation}
where $\phi_{k}(h):\mathbb{N}^+\to\mathbb{R}_{\geq 0}$ is the function defined in
Theorem   \ref{thrm:TreeUniq}.

In the graph $\mathbold{N}'$,  the component of $w_j$, i.e. $N_{j,h}$ is a tree 
and it is only the vertices at distance $h$ from $w_j$ that belong to $\mathbold{B}$.
The colouring of the vertices in $\Gamma$ does not affect the colour assignment
of $w_j$, since we have deleted the edges of $P$ which are incident to $w_j$.
Since ${N}_{j,h} \in {\cal T}(d,\epsilon,h)$,   Proposition  \ref{prop:spatial-monotone}
implies that   \eqref{eq:target-mu'} is indeed true for any  ${N}_{j,h}\in  {\cal T}(d,\epsilon, h)$.

Combining \eqref{eq:target-mu'} and \eqref{eq:MuSigmaVsNuSigma} we get that
\begin{equation}\label{eq:PrXB+Bound}
\left|\Pr[X(v_j)=c \ |\ {N}_{j,h},  X_{\mathbold{B}^+}=\sigma_{\mathbold{B}^+}]- ({k-|t_{\sigma}|)^{-1}} \right|
\leq 10k^{-1}  \phi_k(h).
\end{equation}
The proposition follows from \eqref{eq:PrXB+Bound} and  \eqref{eq:PrXGammaBound} and setting $f_{\epsilon}(h)=10\phi_k(h)$.
\end{propositionproof}

\section{Proof of Proposition \ref{prop:spatial-monotone}}\label{sec:prop:spatial-monotone}

Let $T'$ be a supertree of $T$ such that   $T' $  is a complete  $\Delta$-ary tree of height $h$. 
That is, $T$ and $T'$ have the same height. Also,  both trees have the same root $r$.
We denote with $L$ the set of vertices at level $h$ in $T'$.   $L_0\subseteq L$ is the set 
of vertices which are at level $h$ in both $T$ and $T'$.

For $T$ and $T'$ we have  the following result.
\begin{lemma}\label{lemma:spatial-monotone}
Assume that $k\geq \Delta+2$. Let $X,Y$ be  random $k$-colourings of $T,T'$, respectively.
Also, let $\sigma$ be any $k$-colouring of $T$. 
 For any  $c\in [k]$ it holds that
\[
\Pr\left[X(r)=c \ | \ X(L_0)=\sigma_{L_0}\right]=\Pr\left[Y(r)=c \ | \ Y(L_0)=\sigma_{L_0}\right].
\]
\end{lemma}
The  proof of Lemma \ref{lemma:spatial-monotone} appears in Section \ref{sec:lemma:spatial-monotone}.

Given Lemma \ref{lemma:spatial-monotone},  we show the proposition by working as follows:
Let $X$, $Y$ be a random $k$-colouring of $T$ and $T'$, respectively. 
Let $\tau\in \Omega_{T,k}$ be  such that $\tau_{L_0}$ maximizes the following
quantity,
\[
|\Pr[X(r)=c \ |\ X(L_0)=\tau_{L_0}]-k^{-1} | .
\]
By Lemma \ref{lemma:spatial-monotone}, we have that 
$\Pr[X(r)=c \ |\ X(L_0)=\tau_{L_0}]=\Pr[Y(r)=c \ | \ Y(L_0)=\tau_{L_0}].$ It holds that  
\[
|\Pr[X(r)=c \ |\ X(L_0)=\tau_{L_0}]- k^{-1} |\leq\max_{\sigma\in \Omega_{T',k}}|\Pr[Y(r)=c \ |\ Y(L_0)=\sigma_{L_0}]-k^{-1} |, 
\]
where $\sigma$ varies over all the proper colourings of $T'$.
The proposition  follows by using Theorem \ref{thrm:TreeUniq} to bound the r.h.s. of the inequality above.

\subsection{Proof of Lemma \ref{lemma:spatial-monotone}}\label{sec:lemma:spatial-monotone}

For the tree ${T}$ (resp. the tree $T'$)  and a vertex $v$,  let ${T}_v$
 (resp. $T'_v$) denote the subtree that contains the vertex  $v$ once we delete the edge 
of ${T}$ (resp. $T'$) that connects $v$ and   its parent.
For the tree ${T}_v$ (resp. $T'_v$)  the root is the vertex $v$.

Consider the random colourings $X, Y$ of the trees ${T}$ and $T'$, 
respectively, with boundary condition $\sigma_{L_0}$.  Also, consider the following random 
variables:  For every vertex $v\in {T}$,  (resp. $T'$) we consider the subtree ${T}_v$ 
(resp. $T'_v$) and the random colouring
$X^v$ (resp. $Y^v$) on this tree, with boundary conditions set as follows: Letting 
$L_v=L_0\cap T_v$, then the boundary condition for both $X^v$  and $Y^v$ 
is $\sigma_{L_v}$.

We denote with $C$  the set of the children of the root $r$ which belong to both trees, 
${T},T'$. Also, we denote with $S$ be the set of children of $r$ which belong 
only to the tree $T'$.

The proof is by induction on the height of the tree $h$. We start with $h=1$. 
Since the height of the tree is $1$, it holds that $C=L_0$.
Clearly for any color which appears in the boundary it holds that
neither $X$ nor $Y$ is going to use it for colouring the root.
Let $U\subset [k]$ contain all the colours that are not used by
the boundary condition $\sigma_{L_0}$.
For any $c\in U$ it holds that 
\begin{eqnarray}
\Pr[ Y(r)=c \ |\ Y(L_0)=\sigma_{L_0}]&=&\frac{\prod_{v\in S}(1-\Pr[Y^v(v)= c]) \times \prod_{v\in C}(1-\Pr[Y^v(v)= c])}
{\sum_{q\in [k]}\left( \prod_{v\in S}(1-\Pr[Y^v(v)= q])\times \prod_{v\in C}(1-\Pr[Y^v(v)= q]\right)}\nonumber\\
&=&\frac{\prod_{v\in S}(1-\Pr[Y^v(v) = c]) }
{\sum_{q\in U} \prod_{v\in S}(1-\Pr[Y^v(v)= q]) }. \nonumber
\end{eqnarray}
To see why the second inequality holds consider the following:   
If  $q\notin U$, then we have that $\prod_{v\in C}(1-\Pr[Y^v(v)= q])=0$,    since, 
we have assumed that  there is $v\in C$ such that $Pr[Y^v(v)= q]=1$.
On the other hand, if $q\in U$,  then  $\prod_{v\in C}(1-\Pr[Y^v(v)= q])=1$ 
since, by definition, for every $v\in C$ it holds that $\Pr[Y^v(v)= q]=0$.
Furthermore,  it is direct that
\[
\Pr[ Y(r)=c \ |\ Y(L_0)=\sigma_{L_0}] = \frac{(1-1/k)^{|S|} }
{|U|(1-1/k)^{|S|} } =\frac{1}{|U|}
=\Pr[ X(r)=c \ |\ X(L_0)=\sigma_{L_0}]. 
\]
Assume now that our hypothesis is true for trees of height $h-1$, for some
$h\geq 2$.   We are going to show that the hypothesis is true for trees of height $h$, too.
It holds that 
\begin{eqnarray}
\Pr[ X(r)=c \ |\ X(L_0)=\sigma_{L_0}]&=& 
\frac{\prod_{v\in C}(1-\Pr[X^v(v)= c])}
{\sum_{q\in [k]} \prod_{v\in C}(1-\Pr[X^v(v)= q])}
\nonumber \\
&=& \frac{\prod_{v\in C}(1-\Pr[Y^v(v)= c])}
{\sum_{q\in [k]} \prod_{v\in C}(1-\Pr[Y^v(v)= q])},\label{eq:XviANDYvi}
\end{eqnarray}
where the second equality follows from the induction hypothesis.
Also, it holds that 
\begin{eqnarray}
\Pr[ Y(r)=c \ |\ Y(L_0)=\sigma_{L_0}]&=&\frac{\prod_{v\in S}(1-\Pr[Y^v(v)= c]) \times \prod_{v\in C}(1-\Pr[Y^v(v)= c])}
{\sum_{q\in [k]}\left( \prod_{v\in S}(1-\Pr[Y^v(v)= q])\times \prod_{v\in C}(1-\Pr[Y^v(v)= q]\right)}\nonumber\\
&=&\frac{(1-1/k)^{|S|}\prod_{v\in C}(1-\Pr[Y^v(v) = c]) }
{\sum_{q\in [k]} \left((1-1/k)^{|S|} \prod_{v\in C}(1-\Pr[Y^v(v)= q] \right) }  \nonumber\\
&=&\frac{\prod_{v\in C}(1-\Pr[Y^v(v) = c]) }
{\sum_{q\in [k]} \prod_{v\in C}(1-\Pr[Y^v(v)= q]) },  \label{eq:YviANDYvi}
\end{eqnarray}
where the second equality holds because  for every $v\in S$ it holds $\Pr[Y^v(v)= c]=k^{-1}$. Observe that
if $v\in S$, then the subtree $T'_v$ contains no vertex $u$ which also belongs to
${T}$, thus $Y^v$ has no boundary conditions at all. 
The lemma follows from \eqref{eq:XviANDYvi} and \eqref{eq:YviANDYvi}.

\section{Proof of Theorem \ref{theorem:STEPAccuracy}}\label{sec:theorem:STEPAccuracy}

For proving Theorem \ref{theorem:STEPAccuracy} we need the following result.
\begin{lemma}\label{lemma:isomorphism}
For any $c,q\in [k]$ such that  $c\neq q$, it holds that   ${\tt Switching}(G,v,\cdot, q):S_q(c,c) \to S_c(q,c)$ is 
a {\em bijection}.
\end{lemma}
\begin{proof}
For any $\sigma\in S_q(c,c)$, it holds that ${\tt Switching}(G,v,\sigma,q)\in S_c(q,c)$. 
This follows from Lemma \ref{lemma:SwitchProperColour} and the definition of the sets
$S_q(c,c)$ and $S_c(q,c)$.

It suffices to show that the mapping  ${\tt Switching}(G,v,\cdot,q):S_q(c,c)\to S_c(q,c)$ it is {\em one-to-one}
and it  is {\em surjective}, i.e. it has range $S_c(q,c)$. For showing both properties we  use the following observation:
If for some  $\tau\in S_c(q,c)$ and    $\xi\in S_q(c,c)$ it holds that  $\tau={\tt Switching}(G,v,\xi,  q)$,
then it also holds that $\xi={\tt Switching}(G,v,\tau,c)$.

As far as  surjectiveness is regarded, it suffices to have that  for  every $\tau\in S_c(q,c)$ there exists
 $\xi\in S_q(c,c)$ such that ${\tt Switching}(G,v,\xi,  q)=\tau$.
From the above observation we get that each $\tau\in S_c(q,c)$ is 
the image of $\xi\in S_q(c,c)$ for which it holds that 
$\xi={\tt Switching}(G,v,\tau,c)$.  
Furthermore, we observe that this $\xi$ is unique. This implies that
${\tt Switching}(G,v,\cdot, q)$ is {\em one-to-one}, too.

The lemma follows.
\end{proof}

\begin{theoremproof}{\ref{theorem:STEPAccuracy}}
Let $X,Y$ be the input and the output of ${\tt Update}$,  respectively.  $X$ is distributed uniformly at 
random among the $k$-colourings of $G$.  Also, let $Z$ be a random variable distributed as in
$\nu$, the uniform distribution over the good $k$-colourings of $G$.

The theorem will follow by providing a coupling of $Z$ and $Y$
such that 
\begin{displaymath}
\Pr[Z\neq Y]\leq \alpha.
\end{displaymath}

\noindent
First, we need the following observations:
For any $q,c\in [k]$ such that 
$c\neq q$, it holds that
\begin{equation}
\Pr[Z(v)=q \ | \ Z(u)=c]=\Pr[X(v)=q \ |\ X(u)=c,X(v)\neq c]=(k-1)^{-1}\label{eq:equationA}
\end{equation}
and 
\begin{equation}
\Pr[X(v)=X(u)=c \ |\ \textrm{$X$ is bad }]=k^{-1}.\label{eq:equationB}
\end{equation}
All the above equalities follow due to symmetry between the colours. Also, it is direct to show that
\begin{equation}\label{eq:equationC}
\Pr[Y(v)=q \ | \ X(u)=c]=(k-1)^{-1}.
\end{equation}
In particular, \eqref{eq:equationC} holds because  $Y(v)$ is set according to the following rules:
 if $X$ is good, then we have  that $X=Y$ and   \eqref{eq:equationA} holds. On the other hand, if $X$ is bad and $X(u)=c$, 
then $Y(v)$ is chosen uniformly at random from $[k]\backslash\{c\}$.

Now we are going to describe the coupling.  We need to involve the variable $X$ in the coupling,
since $Y$ depends on it.
At the beginning, we set  $Z(u)=X(u)$,  also we set $Z(v)=Y(v)$. 
From \eqref{eq:equationA},  \eqref{eq:equationB} and \eqref{eq:equationC},
it is direct that  $Z(u)$ and $Z(v)$ are set according to the appropriate distribution.

We need to consider two cases, depending on whether $X$ is a good or a bad colouring.
For each case we have different couplings.
Then it holds that
\begin{equation}\label{eq:DistinguishGoodBad}
\Pr[Y\neq Z]\leq \Pr[Y\neq Z \ | \ \textrm{$X$ is good}]+\Pr[Y\neq Z\ | \ \textrm{$X$ is bad}].
\end{equation}

\noindent
If $X$ is good,  then it is distributed uniformly at random among the good colourings of
$G$. That is, $X$ and $Z$ are identically distributed.   That is, if $X$ is good, then  there is a coupling 
such that $X=Z$  with probability 1. Also, from  ${\tt  Update}$ we have that $X=Y$.   
It is direct that if $X$ is good, then there is a coupling such that 
\begin{equation}\label{eq:CoupleXYZ-XGood}
\Pr[Y\neq Z \ |\ \textrm{$X$ is good}]=0.
\end{equation}

\noindent
On the other hand, if $X$ is a bad colouring, the situation is as follows:
If $X(u)=X(v)=c$, for some $c\in [k]$, then
$Z(u)=c$ and $Z(v)=q$ for some $q\in [k]\backslash\{c\}$
and $Y(v)=q$.
We let the event $E_{c,q}=$ ``$X(u)=X(v)=Z(u)=c$ and $Y(v)=Z(v)=q$
while $X\in S_q(c,c)$ and $Z\in S_c(q,c)$".
Also,  let the event $E=\bigcup_{c,q\in [k]:c\neq q}E_{c,q}$.

In the coupling we are distinguishing the cases where the event $E$
occurs from those that is does  not. For each case we have different couplings.  It holds that
\begin{equation}\label{eq:CouplingBadBound}
\Pr[Y\neq Z \ |\ \textrm{$X$ is bad}]\leq \Pr[Y\neq Z\ |\ E, X \textrm{ is bad}]+\Pr[\bar E\ |\ X \textrm{ is bad}],
\end{equation}
where $\bar E$ is the complement of $E$.  The theorem  follows by showing that the r.h.s. of 
\eqref{eq:CouplingBadBound}  is  at most $\alpha$. From the  definition of the quantity $\alpha$
(Definition \ref{assm:isomorphism}),   it holds that 
\[
\Pr[X\in S_q(c,c) \ |\ X(u)=X(v)=c]\geq 1-\alpha,
\]
also, it holds that 
\[
\Pr[Z\in S_c(q,c)\ |\ Z(u)=c, Z(v)=q]\geq 1-\alpha, 
\]
for any $c,q\in [k]$ and $q\neq c$.  The  above implies that,  when $X$ is bad, there is a coupling 
such that 
\begin{equation}\label{eq:PrE}
\Pr[E\ |\  X\textrm{ is  bad}]\geq 1-\alpha.
\end{equation}
It remains to describe a coupling of $Z,Y$, when $X$ is bad and $E$ occurs 
(i.e. bound $\Pr[Y\neq Z\ |\ E, X \textrm{ is bad}]$).  For this, we need the following claim.

\begin{myclaim}\label{claim:YUniformE}
Conditional on the event $E_{c,q}$, $Y$ is distributed uniformly over
$S_c(q,c)$.
\end{myclaim}
\begin{proof}
From Lemma \ref{lemma:isomorphism}  we have that ${\tt Switching}(G,v,\cdot, q):S_q(c,c) \to S_c(q,c)$
is a bijection. The existence of this bijection implies that $|S_q(c,c)|=|S_c(q,c)|$. 
Also, for each $\tau\in S_c(q,c)$ there is a unique $\xi \in S_q(c,c)$ such that ${\tt Switching}(G,v,\xi, q)=\tau$.
Clearly $\Pr[Y=\tau\ |\ E_{c,q}]=\Pr[X=\xi \ |\  E_{c,q}]$.

Conditional on the event $E_{c,q}$, the random variable $X$ is distributed
uniformly over $S_q(c,c)$.  Thus, $\Pr[Y=\tau|E_{c,q}]=|S_q(c,c)|^{-1}=|S_c(q,c)|^{-1}$, for any
$\tau \in S_c(q,c)$. The claim follows.
\end{proof}

\noindent
It is direct that conditional on $E_{c,q}$ the random variable $Z$ is distributed
uniformly at random in $S_c(q,c)$.  Also, observe that conditional on that $X$ is bad  and 
$E$ occurring,  we are going to have $Z(v)=Y(v)$ and $Z(u)=Y(u)$.  All these imply that 
there is a coupling of $Z,Y$ such that
\begin{equation}\label{eq:EXbadBound}
\Pr[Y\neq Z \ |\ X\text{ is bad}, E]=0.
\end{equation}
Plugging \eqref{eq:PrE} and \eqref{eq:EXbadBound} into \eqref{eq:CouplingBadBound}, we
get that 
\[
\Pr[Y\neq Z \ |\ \textrm{$X$ is bad}]\leq \alpha.
\]
The theorem follows by plugging the above bound and \eqref{eq:CoupleXYZ-XGood} into 
\eqref{eq:DistinguishGoodBad}.
\end{theoremproof}

\section{The rest of the proofs}

\subsection{Lemma \ref{lemma:SwitchProperColour}}\label{sec:lemma:SwitchProperColour}
We show that for any $\sigma \in \Omega_{G,k}$, it holds that ${\tt Switching}(G,v,\sigma,q)$ 
returns a   proper colouring of $G$. Assume the contrary, i.e.  there is  $\sigma\in\Omega_{G,k}$ such
that for $\tau={\tt Switching}(G,v,\sigma, q)$ it holds that $\tau\notin \Omega_{G,k}$.

Let the disagreement graph $\mathbold{Q}=\mathbold{Q}(G,v,\sigma, q)$. 
Since $\tau$ is non-proper is has at least one monochromatic edge.
The  monochromatic edge can be  incident either to two vertices in $\mathbold{Q}$ 
or  to  some vertex in $\mathbold{Q}$ and some vertex outside $\mathbold{Q}$.
We are going to show that neither of the  two cases can happen.

${\tt Switching}(G,v,\sigma,q)$ cannot create any monochromatic edge between two vertices in $\mathbold{Q}$. 
To see this, note that the disagreement graph $\mathbold{Q}$ is bipartite and $\sigma$ specifies exactly one colour 
for each part of the graph.  ${\tt Switching}(G,v,\sigma,q)$ just exchanges  the colours of the two parts in the graph.
Clearly this operation cannot generate a monochromatic  of the first kind. 

${\tt Switching}(G,v,\sigma,q)$   cannot cause any monochromatic edge between a vertex in 
$\mathbold{Q}$ and some vertex outside $\mathbold{Q}$, either.
This follows by the fact that the disagreement graph is maximal. That is, there is no vertex 
$w$ outside $\mathbold{Q}$  such that  $\sigma_w\in \{q,c\}$ while at the same 
time $w$ is adjacent to some vertex in $\mathbold{Q}$. Since the recolouring that
${\tt Switching}(G,v,\sigma,q)$ does, involves only vertices coloured  $c,q$,  no
monochromatic edge of the second kind can be generated, too.

The lemma follows.

\subsection{Lemma \ref{lemma:StepAc}}\label{sec:lemma:StepAc}

The time complexity of computing   ${\tt Switching}(G,v, \sigma, q)$
is dominated by the time we need to reveal the disagreement graph 
$\mathbold{Q}=\mathbold{Q}(G,v,\sigma,q)$. We will  show that
we need $O(|E(G)|)$  steps to get $\mathbold{Q}$.

We reveal  the graph $\mathbold{Q}$ in steps $j=0,\ldots, |E(G)|$.
At step $0$, we have  $\mathbold{Q}(0)$ which contains only the vertex $v$. 
Given  $\mathbold{Q}(j)$ we construct $\mathbold{Q}(j+1)$ 
as follows: Pick some edge which is incident to a vertex in $\mathbold{Q}(j)$. 
If the other end of this edge  is incident to a vertex outside $\mathbold{Q}(j)$ that is coloured
either $\sigma_v$ or $q$, then we get $\mathbold{Q}(j+1)$ by inserting
this edge and the vertex into $\mathbold{Q}(j)$. Otherwise 
$\mathbold{Q}(j+1)$ is the same as $\mathbold{Q}(j)$. 
We never pick the same edge twice in the process above.

 The lemma follows by noting that  the process has at most $|E|$ steps,
 while at the end   we get $\mathbold{Q}$.

\subsection{Theorem \ref{thrm:Accuracy}}\label{sec:thrm:Accuracy}

For $i=0,\ldots, r$ consider the following:
Let $\mu_i$ denote the uniform distribution over the $k$-colourings of $G_i$.
Also let $\hat{\mu}_i$ denote the distribution of $Y_i$, where $Y_i$ is the colouring
that the algorithm assigns to the graph $G_i$.  Finally, let $\nu_i$ denote the 
distribution of the output colouring of ${\tt Update}(G_i,v_i, u_i,X_i,k)$ where $X_i$
is distributed as in $\mu_i$.

The theorem follows by showing that that
\begin{equation}\label{eq:thrm:Accuracy:base}
|| \mu_r-\hat{\mu}_r||\leq \sum^{r-1}_{i=0}\alpha_i.
\end{equation}
Theorem \ref{theorem:STEPAccuracy} implies the following:
For every $i=1,\ldots, r$ it holds that
\begin{equation}\label{eq:STEPAccuracyRel}
|| \mu_i - \nu_{i-1}|| \leq \alpha_{i-1},
\end{equation}
It suffices to show that 
\begin{equation}\label{eq:TriangleDecom}
|| \mu_r-\hat{\mu}_r|| \leq \sum^{r}_{i=1}||\mu_i- \nu_{i-1} ||,
\end{equation}
since  it is direct that \eqref{eq:thrm:Accuracy:base} follows from \eqref{eq:STEPAccuracyRel}
and \eqref{eq:TriangleDecom}.

For getting \eqref{eq:TriangleDecom},  we are going to show  for any $i=1,\ldots, r$ the following is true:
\begin{equation}\label{eq:Sufficient4eq:TriangleDecom}
|| \nu_{i-1}-\hat{\mu}_i|| \leq ||\mu_{i-1}-\hat{\mu}_{i-1} ||.
\end{equation}
From \eqref{eq:Sufficient4eq:TriangleDecom} we get to \eqref{eq:TriangleDecom}
by working as follows:  Using the triangle inequality, we have that
\begin{eqnarray}
|| \mu_r-\hat{\mu}_r|| &\leq & || \mu_r- \nu_{r-1}||+|| \nu_{r-1}-\hat{\mu}_r|| \nonumber\\
&\leq& || \mu_r- \nu_{r-1}||+|| \mu_{r-1}-\hat{\mu}_{r-1}||. \qquad\mbox{[from \eqref{eq:Sufficient4eq:TriangleDecom}]} \nonumber
\end{eqnarray}
We work with the term $|| \mu_{r-1}-\hat{\mu}_{r-1}||$, above, in the same way as we did with
$|| \mu_r-\hat{\mu}_r||$ and so on.  This sequence of substitutions and the fact that
$||\mu_0-\hat{\mu}_0 ||=0$,  yield  \eqref{eq:TriangleDecom}.

It remains to show \eqref{eq:Sufficient4eq:TriangleDecom}.
For this, let  $X_{i-1}$ be a random $k$-colouring of the graph $G_{i-1}$ and let
$Z_{i}={\tt Update}(G_{i-1},v_{i-1},u_{i-1},X_{i-1},k).$ It is direct that $Z_i$ is distributed 
as in $\nu_{i-1}$.
Let $Y_{i-1}, Y_i$ be the colouring that  the algorithm assigns to the graphs $G_{i-1}$, $G_i$, 
respectively. Clearly it holds that $Y_i={\tt Update}(G_{i-1},v_{i-1},u_{i-1}, Y_{i-1},k)$

So as to bound $|| \nu_{i-1}-\hat{\mu}_i||$ we consider the following coupling 
of $Z_i$ and $Y_i$: We couple $X_{i-1}$ and $Y_{i-1}$ {\em optimally}. Then  from $X_{i-1}$
and $Y_{i-1}$,  we get $Z_i$ and $Y_{i}$, respectively, as described above. 
By the coupling lemma we have the following
\begin{equation}
|| \nu_{i-1}-\hat{\mu}_i||  \leq \Pr[Z_i\neq Y_i] 
\leq \Pr[Z_i\neq Y_i \ | \  X_{i-1}=Y_{i-1}] + \Pr[X_{i-1}\neq Y_{i-1}]. \label{eq:1821A}
\end{equation}
It is direct that if $X_{i-1}=Y_{i-1}$, then there is a coupling which yield $Z_i= Y_i$ with probability 1.
That is, $\Pr[Z_i\neq Y_i \ |\  X_{i-1}=Y_{i-1}]=0$. Also,  since we have coupled  $X_{i-1}$ and $Y_{i-1}$
optimally,  it holds that 
\begin{equation}\label{eq:OptCouplXY}
\Pr[X_{i-1}\neq Y_{i-1}]=||\mu_{i-1}-\hat{\mu}_{i-1} ||.
\end{equation}
Plugging \eqref{eq:OptCouplXY} into \eqref{eq:1821A} and using the fact that 
$\Pr[Z_i\neq Y_i \ |\ X_{i-1}=Y_{i-1}]=0$, we get \eqref{eq:Sufficient4eq:TriangleDecom}.
The theorem follows.

\subsection{Lemma \ref{lemma:seqsubprop}}\label{sec:lemma:seqsubprop}
 It suffices to show that with probability at least $1-n^{-2/3}$ for any two cycles in $\G$, of 
 maximum length  $(\log_d n)/9$  do not share edges and vertices with each 
other.  Assume the opposite, i.e. that there are at least two  such cycles that intersect with 
each other. Then, there  must exist a subgraph of $\G$ that contains at most $(2/9) \log_d n$ 
vertices  while the number of edges exceeds  by 1, or more, the number of vertices.

Let $D$ be the event that in $\G$ there exists a set of $r$ vertices which have 
$r+1$ edges between them, for $r \leq (2\log_d n)/9$.  The lemma follows
by showing that $\Pr[D]\leq n^{-2/3}$.

We have the following:
\begin{displaymath}
\begin{array}{lcl}
\Pr[D] & \leq & \displaystyle \sum_{r=1}^{(2/9)\log_d n} {n \choose r} { {r \choose 2} \choose r+1} (d/n)^{r+1} (1-d/n)^{{r \choose 2}-(r+1)}
\\ \vspace{-.3cm} \\
& \leq & \displaystyle \sum_{r=1}^{(2/9) \log_d n} \left( \frac{n e}{r} \right )^r \left ( \frac{r^2 e}{2(r+1)} \right )^{r+1} (d/n)^{r+1} 
\leq 
\frac{e \cdot d}{2n}\sum_{r=1}^{(2/9) \log_d n} r\left ( \frac{e^2 d}{2} \right )^{r}
\\ \vspace{-.3cm}\\
& \leq & 
\displaystyle \frac{C\log n}{n} \left (\frac{e^2d}{2} \right)^{(2/9)   \log_d n}.
\end{array}
\end{displaymath}
Let $\gamma=\frac{2  \log(e^2d/2)}{9\log d}$. The quantity in the r.h.s. of the last inequality, above,  
 is of order $\Theta(n^{\gamma-1}\log n)$.  
Taking large $d$ it holds that $\gamma < 0.25 $. Consequently,  we get that  $\Pr[D]\leq n^{-2/3}$.
The lemma follows.

\paragraph{Acknowledgement.} The author of this work would like to thank Amin Coja-Oghlan 
and Elchanan Mossel for the fruitful discussions we had. Also, I would like to thank the anonymous
reviewers for helping me  improved the content of the paper.

\end{document}